\pdfoutput=1

\documentclass[letterpaper,journal]{IEEEtran}

\ifCLASSINFOpdf
   \usepackage[pdftex]{graphicx}
   \graphicspath{{./}{./img/}}
   \DeclareGraphicsExtensions{.pdf,.jpeg,.jpg,.png,.eps,.ps}
\else
  \usepackage[dvips]{graphicx}
  \graphicspath{{./}{./img/}{./eps/}}
  \DeclareGraphicsExtensions{.pdf,.jpeg,.jpg,.png,.eps,.ps}
\fi
\usepackage{cite}
\usepackage[cmex10]{amsmath}
\usepackage{algorithmic}
\usepackage{array}
\usepackage{mdwmath}
\usepackage{mdwtab}
\usepackage{eqparbox}
\usepackage[font=footnotesize]{subfig}
\usepackage{url}

\usepackage[lined, ruled, linesnumbered,noend]{algorithm2e}
\usepackage[bookmarks=false]{hyperref}
\usepackage{latexsym}
\usepackage{amsfonts,amssymb,amsthm}
\usepackage[usenames,dvipsnames,svgnames,table]{xcolor}

\usepackage{enumitem}
\usepackage{pbox}
\usepackage{mathtools}
\usepackage{float}
\usepackage{pgffor}
\usepackage{bm}		
\usepackage{tabularx}	
\usepackage{framed}	
\usepackage{cleveref}	
\usepackage{pifont}		
\usepackage[top=0.763in, bottom=0.68in, left=0.68in, right=0.6in]{geometry}
\IEEEoverridecommandlockouts

\usepackage[export]{adjustbox}

\IEEEaftertitletext{\vspace{-1em}}

\newtheorem{definition}{\bf Definition}	%
\newtheorem{theorem}{\bf Theorem}		%
\newtheorem{lemma}{\bf Lemma}		%
\newtheorem{proposition}{\bf Proposition}		%
\newcounter{appdx}
\newcommand{\kETAL}     {{\em et~al.}}		%
\newcommand{\kIE}     {{\em i.e.}}		%

\newcommand{\smallfont}  {}

\definecolor{dgreen}{rgb}{0,0.655,0.149}

\newcommand{\redc}[1]{{\color{red}}}

\SetKw{Continue}{continue}
\SetKw{Break}{break}

\newcommand{\myendbox}{\hfill \qed}
\makeatletter
\newcommand{\nosemic}{\renewcommand{\@endalgocfline}{\relax}}%
\newcommand{\dosemic}{\renewcommand{\@endalgocfline}{\algocf@endline}}%
\let\oldnl\nl%
\newcommand{\nonl}{\renewcommand{\nl}{\let\nl\oldnl}}%
\makeatother

\SetKwInOut{Init}{Initialize}
\SetKwInOut{OptIn}{Optional input}

\begin{document}


\title{FENDI: Toward High-\underline{F}idelity \underline{En}tanglement \underline{Di}stribution in the Quantum Internet%
\vspace{-2mm}
\author{Huayue~Gu$^{\ddagger}$,~\IEEEmembership{Student Member,~IEEE,}
        Zhouyu~Li$^{\ddagger}$,~\IEEEmembership{Student Member,~IEEE,}
        Ruozhou~Yu,~\IEEEmembership{Senior Member,~IEEE,}\\
        Xiaojian~Wang,~\IEEEmembership{Student Member,~IEEE,}
        Fangtong~Zhou,~\IEEEmembership{Student Member,~IEEE,}\\
        Jianqing~Liu,~\IEEEmembership{Member,~IEEE,}
        and~Guoliang~Xue,~\IEEEmembership{Fellow,~IEEE}}
\IEEEcompsocitemizethanks{\IEEEcompsocthanksitem
$^{\ddagger}$ Both authors contributed equally to this research.
\IEEEcompsocthanksitem Gu, Li, Yu, Wang, Zhou and Liu (\{hgu5, zli85, ryu5, xwang244, fzhou, jliu96\}@ncsu.edu) are with NC State University, Raleigh, NC 27606, USA. Xue (xue@asu.edu) is with Arizona State University, Tempe, AZ 85281, USA.
\IEEEcompsocthanksitem Gu, Li, Yu, Wang, Zhou were supported in part by NSF grant 2045539. 
\IEEEcompsocthanksitem Liu was supported in part by NSF grants 2304118 and 2326746.
\IEEEcompsocthanksitem Xue was supported in part by
NSF grants 2007083, 2007469,
and by the PiQPsi project of Advanced Scientific Computing Research program,
U.S. Department of Energy under FWP No. ERKJ432.
\IEEEcompsocthanksitem The information reported herein does not reflect the position or the policy
of the funding agencies.
}
}

\markboth{Manuscript Currently Under Review at IEEE/ACM TRANSACTIONS ON NETWORKING}%
{How to Use the IEEEtran \LaTeX \ Templates}
%
%
%
%

%
%

%
%
%

%


%
\IEEEtitleabstractindextext{
\begin{abstract}
A quantum network distributes quantum entanglements between remote nodes, and is key to many applications in secure communication, quantum sensing and distributed quantum computing.
This paper explores the fundamental trade-off between the throughput and the quality of entanglement distribution in a multi-hop quantum repeater network.
Compared to existing work which aims to heuristically maximize the entanglement distribution rate (EDR) and/or entanglement fidelity, our goal is to characterize the maximum achievable worst-case fidelity, while satisfying a bound on the maximum achievable expected EDR between an arbitrary pair of quantum nodes.
This characterization will provide fundamental bounds on the achievable performance region of a quantum network, which can assist with the design of quantum network topology, protocols and applications.
However, the task is highly non-trivial and is NP-hard as we shall prove.
Our main contribution is a \emph{fully polynomial-time approximation scheme} to approximate the achievable worst-case fidelity subject to a strict expected EDR bound, combining an optimal fidelity-agnostic EDR-maximizing formulation and a worst-case isotropic noise model.
The EDR and fidelity guarantees can be implemented by a post-selection-and-storage protocol with quantum memories.
By developing a discrete-time quantum network simulator, we conduct simulations to show the characterized performance region (the approximate Pareto frontier) of a network, and demonstrate that the designed protocol can achieve the performance region while existing protocols exhibit a substantial gap.
\end{abstract}

}
\maketitle
\IEEEdisplaynontitleabstractindextext

\IEEEpeerreviewmaketitle



\vspace{-1em}
\ifCLASSOPTIONcompsoc
\IEEEraisesectionheading{\section{Introduction}\label{sec:intro}}
\else
\section{Introduction}
\label{sec:intro}
\noindent
\fi
\IEEEPARstart{A} quantum network enables efficient quantum communication based on the principle of quantum entanglement~\cite{koashi1998no}. The ability to transmit quantum information between remote nodes is key to many astonishing quantum applications, such as quantum secure communication~\cite{bennett2020quantum}, distributed quantum computing~\cite{cacciapuoti2019quantum,cicconetti2022resource}, and quantum sensor network~\cite{Xia2020}.

While the concept has been proposed for years, practical quantum networking has only come around the corner with recent real-world implementations~\cite{elliott2002building,peev2009secoqc,sasaki2011field,yin2017satellite,dahlberg2019link}.
Though current systems are built in ideal conditions and small-scale in nature, research has looked into how such small-scaled networks could possibly be extended to a fully-fledged, global-scale quantum internet~\cite{zhao2021redundant}.
A key functionality of a quantum internet is to distribute entangled quantum states between remote nodes across long distances.
For instance, an entangled pair of photons can teleport one quantum bit (qubit) between a pair of nodes that are arbitrarily far away from each other.
Future applications would require a steady stream of high-quality entanglements between arbitrary remote ends.

This paper considers a first-generation quantum network built with quantum repeaters~\cite{Muralidharan2016}, which performs entanglement distribution via \emph{entanglement generation} and \emph{entanglement swapping}.
If a quantum link connects a pair of repeaters, a remote entanglement between them can be generated by preparing a pair of entangled photons at an entanglement source, and directly sending each photon to one node.
Entanglements generated over multiple links can further be swapped at joint intermediate nodes to entangle qubits at indirectly connected nodes. %
This way, each end-to-end entanglement is generated along an \emph{entanglement path} in a quantum network.

As entanglements are a critical resource, attention has been drawn to the design of efficient entanglement distribution protocols to ensure the quantity (\emph{aka} entanglement distribution rate or EDR) of entanglement distribution.
A quantum network has unique characteristics imposed by the underlying physics or technology deficits.
First, entanglement distribution efficiency is fundamentally limited by transmission loss of entangled photons, and failures in entanglement swapping.
To mitigate these, many existing works have studied \emph{efficient entanglement routing} to find entanglement paths with maximum success probability~\cite{zhao2021redundant,Dai2020a,shi2020concurrent}.
Second, uncontrollable noise and operation errors can degrade the quality (\emph{aka} fidelity) of distributed entanglements.
Low fidelity results in low communication efficiency due to excessive error correction needed, even when the EDR is high.
Thus when considering entanglement distribution to support various applications, it is essential to consider both EDR and fidelity.

This paper explores the tradeoff between the achievable EDR and fidelity of a quantum network with a general topology.
We start with characterizing the end-to-end fidelity of entanglements distributed over an entanglement path.
Combining it with a recently proposed achievable upper bound of the expected EDR between a pair of nodes, we propose a novel decomposition theorem that is essential for characterizing both the achievable EDR and fidelity between a pair of nodes.
As a next step, we formulate the problem of computing the maximum achievable worst-case fidelity, while trying to satisfy a lower bound on the achievable expected EDR.
This bi-criteria formulation can be used to optimize for many applications, which desire a steady entanglement rate and can benefit from improved end-to-end fidelity.
Our proposed solution, named \textbf{FENDI}, is a \emph{fully polynomial-time approximation scheme} to the formulated bi-criteria problem, which we prove to be NP-hard.
We further show that the computed solution can be implemented with a post-selection-and-storage protocol to achieve both the expected EDR and the end-to-end fidelity.
With the help of discrete event simulation, we demonstrate that FENDI can be used to approximate the EDR-fidelity Pareto frontier of a network efficiently, and show that existing algorithms exhibit a substantial gap from the approximate frontier that can be achieved by the post-selection-and-storage protocol.
Our main contributions are summarized as follows:
\vspace{-0.25em}
\begin{enumerate}
    \item We model a general quantum network with Werner states, and derive an end-to-end fidelity parameter as a product of link and node attributes based on isotropic noise.
    \item We prove a novel decomposition theorem, enabled by a new \emph{primitive entanglement flow (pflow)} abstraction, to characterize the worst-case end-to-end fidelity of entanglement distribution with optimal EDR.
    \item Based on the above, we formulate a bi-criteria problem called \emph{high-fidelity remote entanglement distribution (HF-RED)} between a pair of nodes, and prove it is NP-hard.
    \item We propose a \emph{fully polynomial-time approximation scheme (FPTAS)} to maximize the worst-case end-to-end fidelity subject a lower bound on the expected EDR, and realize the solution with a post-selection-and-storage protocol.
    \item We develop a discrete event quantum network simulator implementing the protocol, characterize the (approximate) EDR-fidelity frontier, and compare existing protocols to the post-selection-and-storage protocol.
\end{enumerate}
\vspace{-0.25em}

\noindent\textbf{Organization:}
\S\ref{sec:rw} reviews background and related work.
\S\ref{sec:model} introduces our quantum network model.
\S\ref{sec:abs} presents our new abstraction and our decomposition theorem for characterizing the EDR-fidelity trade-off, formulate the HF-RED problem, and show its NP-hardness.
\S\ref{sec:fptas} presents our approximation scheme, analysis and discussion.
\S\ref{pe} presents implementation and simulation results.
\S\ref{sec:conclusions} concludes the paper.
%



\section{Background and Related Work}
\label{sec:rw}
\noindent
The promises of quantum communication advantages in many practical applications (such as sensing, communication and computing) have attracted huge attention across the globe~\cite{yin2017satellite,esdi,liao2017satellite}.
The idea of a quantum network was first proposed by the DARPA quantum network project \cite{elliott2002building}.
Early work in quantum networking focused on feasibility demonstration in ideal situations.
Much of the literature has derived analytical and simulation models for quantum repeater chains~\cite{feynman1982simulating,caleffi2018quantum} and other specialized topologies including lattices~\cite{pant2019routing}, star~\cite{vardoyan2021stochastic} and ring-like topologies~\cite{schoute2016shortcuts,chakraborty2019distributed}.
In reality, a quantum internet is unlikely to have such ideal topologies due to physical and geographical limitations.

Recent studies have focused on \emph{entanglement routing} in general quantum networks~\cite{Dai2020b,chang2022order,kaipingxue_2022}.
A common approach was to find paths with highest success probability using modified shortest path algorithms~\cite{van2013designing}.
Shi~\kETAL{}~\cite{shi2020concurrent} first showed that maximum-success paths do not lead to the highest throughput, and proposed algorithms QCAST and QPASS with optimal single-path routing metrics.
Zhao~\kETAL{}~\cite{zhao2021redundant} proposed an algorithm to achieve higher throughput by provisioning redundant intermediate entanglements for swapping.
Zeng~\kETAL{}~\cite{zeng2022multi} proposed an integer programming-based solution using branch-and-price with very limited quantum memories.
Dai~\kETAL{}~\cite{Dai2020a,Dai2020b} proposed the first optimal remote entanglement distribution (ORED) protocol for end-to-end EDR maximization, giving an upper bound on the achievable expected EDR between a pair of nodes in an arbitrary network. 
The above works only considered the success probability but ignored the quality (fidelity) of entanglements.
%

To enable high-quality quantum communication, 
some works have focused on ensuring or improving fidelity.
Zhao~\kETAL{}~\cite{zhao2022e2e} first studied fidelity-aware entanglement routing.
They derived an end-to-end fidelity model based on \emph{bit flip errors} and proposed a purification-based fidelity-aware routing algorithm with heuristic path selection, linear programming, and rounding.
Pouryousef~\kETAL{}~\cite{pouryousef2022quantum} proposed a quantum overlay network architecture, utilizing entanglement purification to maximize the weighted entanglement generation rate for multiple users.
Panigrahy~\kETAL{}~\cite{panigrahy2022capacity} also proposed a max-weight scheduling policy and proved its stability for all arrival requests in a star-shaped network topology.
However, these studies are based on strong assumptions and constraints, and do not provide theoretical guarantee on the achievable EDR and fidelity region of a general network.
Our study on theoretical guarantees for characterizing the EDR-fidelity trade-off in a quantum network is motivated by the above works and their limitations.

%
%
%
%
%
%



%

\section{System Model}
\label{sec:model}
\noindent
In this section, we present preliminaries of a quantum network. Notations related to modeling are summarized in Table~\ref{notation}.

\subsection{Quantum Basics}
\noindent
Consider a common $2$-state quantum system with orthonormal basis states $|0\rangle$ and $|1\rangle$.
A quantum bit (\emph{qubit}) is a superposition of $|0\rangle$ and $|1\rangle$, written as $|b\rangle\! = \!\alpha |0\rangle + \beta |1\rangle$, satisfying $|\alpha|^2 + |\beta|^2 = 1$.
A perfect measurement on $|b\rangle$ yields classical state $0$ with probability $|\alpha|^2$ and $1$ with probability $|\beta|^2$.
A two-qubit system is a superposition of four basis states $|00 \rangle$, $|01 \rangle$, $|10 \rangle$ and $|11 \rangle$.
Let $|b_1b_2\rangle\!\! = \!\alpha_{00} |00\rangle \!+ \!\alpha_{01}|01\rangle +\! \alpha_{10}|10\rangle + \alpha_{11}|11\rangle$, such that $|\alpha_{00}|^2 + |\alpha_{01}|^2 + |\alpha_{10}|^2 + |\alpha_{11}|^2 = 1$.
Simultaneous measurement on the two qubits will yield $00$, $01$, $10$ and $11$ with probabilities $|\alpha_{00}|^2$, $|\alpha_{01}|^2$, $|\alpha_{10}|^2$ and $|\alpha_{11}|^2$ respectively.

A maximally entangled pair (Bell pair) is a two-qubit system in one of the four \emph{Bell states}:
$|\Phi^{\pm}\rangle = \frac{1}{\sqrt{2}}( |00\rangle \pm |11\rangle )$, and $|\Psi^{\pm}\rangle = \frac{1}{\sqrt{2}}( |01\rangle \pm |10\rangle )$.
A Bell pair is \emph{maximally entangled} since it only contains two of the four basic states with equal probability, where in both states the two qubits are perfectly correlated.
For instance, in state $|\Psi^+\rangle$, if one of the qubits measures into $x$, then the other must measure into $(1-x)$, for $x \in \{ 0, 1\}$.
Bell pairs (also called \emph{\textbf{ebits}}) form the basis of two-party quantum communications: if Alice and Bob each holds one of two entangled qubits, they can use this pair to send any single-qubit quantum state via \emph{local operations and classical communication (LOCC)}.
Bell pairs can also be used to construct arbitrary multipartite entangled states needed by applications such as distributed quantum sensing~\cite{Xia2020}.

\begin{table}
\caption{Key Notations in Modeling}
\label{notation}
\footnotesize
\begin{tabular}{p{1.7cm}p{6.4cm}}
\hline
Parameters & Description\\
\hline
$G=(N, L)$ & quantum network with nodes $N$ and links $L$\\
$c_{l},F_l$ & capacity and fidelity of link $l$ \\ 
$W_l, W_n$ & fidelity loss parameters of link $l$ and node $n$\\
$q_l, q_n$ & ebit generation \& swapping success probabilities\\
$F^{\sf E2E},P_{st}^{\sf E2E}$ & the end to end fidelity and success probability\\
$\mathcal{G} = (\mathcal{V}, \mathcal{E})$ & induced graph of an eflow or pflow \\
$\eta_{st}$ & the expected EDR between SD $st$\\
$m{{{{{}}}}}n$ & an enode, \kIE{}, an unordered pair of nodes $m, n \in N$\\
$\Psi_{st}$ & the set of all possible pflows between $s$ and $t$ \\
$\Delta_{st}$ & expected EDR bound between $s$ and $t$\\
$\Upsilon_{st}$ & end-to-end fidelity bound between $s$ and $t$\\
$\mathcal{P}_{st}$ & the set of $st$-pflows with fidelity no lower than $\Upsilon_{st}$\\
\hline
Variables & Description\\
\hline
$g_{m{{{{{}}}}}n}$ & elementary ebit generation rate along link ${m{{{{{}}}}}n} \in L$ divided by the capacity $c_{mn}$ \\
$f^{m{{{{{}}}}}k}_{m{{{{{}}}}}n}$ &  rate of ${m{{{{{}}}}}k}$-ebits to be swapped to generate ${m{{{{{}}}}}n}$-ebits\\
$I(mn)$ & total ebit rate generated between node pair $m{{{{{}}}}}n$ \\ 
$\Omega(mn)$ & total ebit rate contributed by $m{{{{{}}}}}n$ to swapping\\
$\eta(\psi)$ & the pflow value (expected EDR) assigned to $\psi \in \Psi_{st}$\\
\hline

\hline
\end{tabular}

\end{table}

\subsection{Quantum Operations}
\noindent
Quantum operations and their characteristics crucially differentiate quantum networking from classical networking.

\textbf{Entanglement generation:} 
Quantum network mainly relies on the generation and transmission of photonic entangled states.
A pair of entangled photons is first generated by a physical process such as \emph{spontaneous parametric down-conversion (SPDC)} at an entanglement source.
Then, both photons are transmitted to two nearby nodes via a quantum link\footnotemark{}.
\footnotetext{Alternatively, the entanglement source can be placed at a repeater, then only one photon needs to traverse the link to the other repeater.}%
The photons can be transmitted via different types of links---optical fiber, free space, or an optical switch network---but suffer from transmission loss that is commonly exponential to the distance traversed~\cite{singh2021quantum, photon-loss}.
We consider generating an entangled photon pair and transmitting one/both photons jointly as the \emph{entanglement generation} process.
Entanglements generated via this process is called \emph{elementary ebits}.

Notably, this is a \emph{probabilistic} process because of both the generation process with non-linear optics and the probabilistic transmission loss.
A heralding and post-selection process is commonly employed after this process to detect successfully entangled and transmitted pairs, and the process can be repeated for many times until one entangled pair is generated.
\textbf{Entanglement swapping:} Considering photon loss during transmission, entanglement swapping via quantum repeaters is essential for long-distance entanglement distribution.
An entanglement swap takes as input two remote entangled pairs---each with one photon on a shared repeater node.
The repeater first entangles the two local photons at the repeater, performs a Bell state measurement (BSM) on the two photons, and then sends the measurement result to either of the two remote nodes via classical communication.
The node receiving the result then performs a local unitary operation on its own qubit, and the two remote photons become entangled without physical interaction.

Similar to generation, swapping is also \emph{probabilistic} with near-term devices. 
Fundamentally, BSM with linear optics can only succeed with no more than $50\%$ probability, since two of the four Bell states are not distinguishable~\cite{bayerbach2023bell}.
Hence when the measurement result matches either of the two indistinguishable states, the remote qubits must be discarded.
Device deficits may further degrade the success probability.

Fig.~\ref{fig:swap} illustrates the process of distributing an entanglement between two nodes connected by a quantum repeater.
\begin{figure}[t!]
\centering
\includegraphics[width=0.35\textwidth]{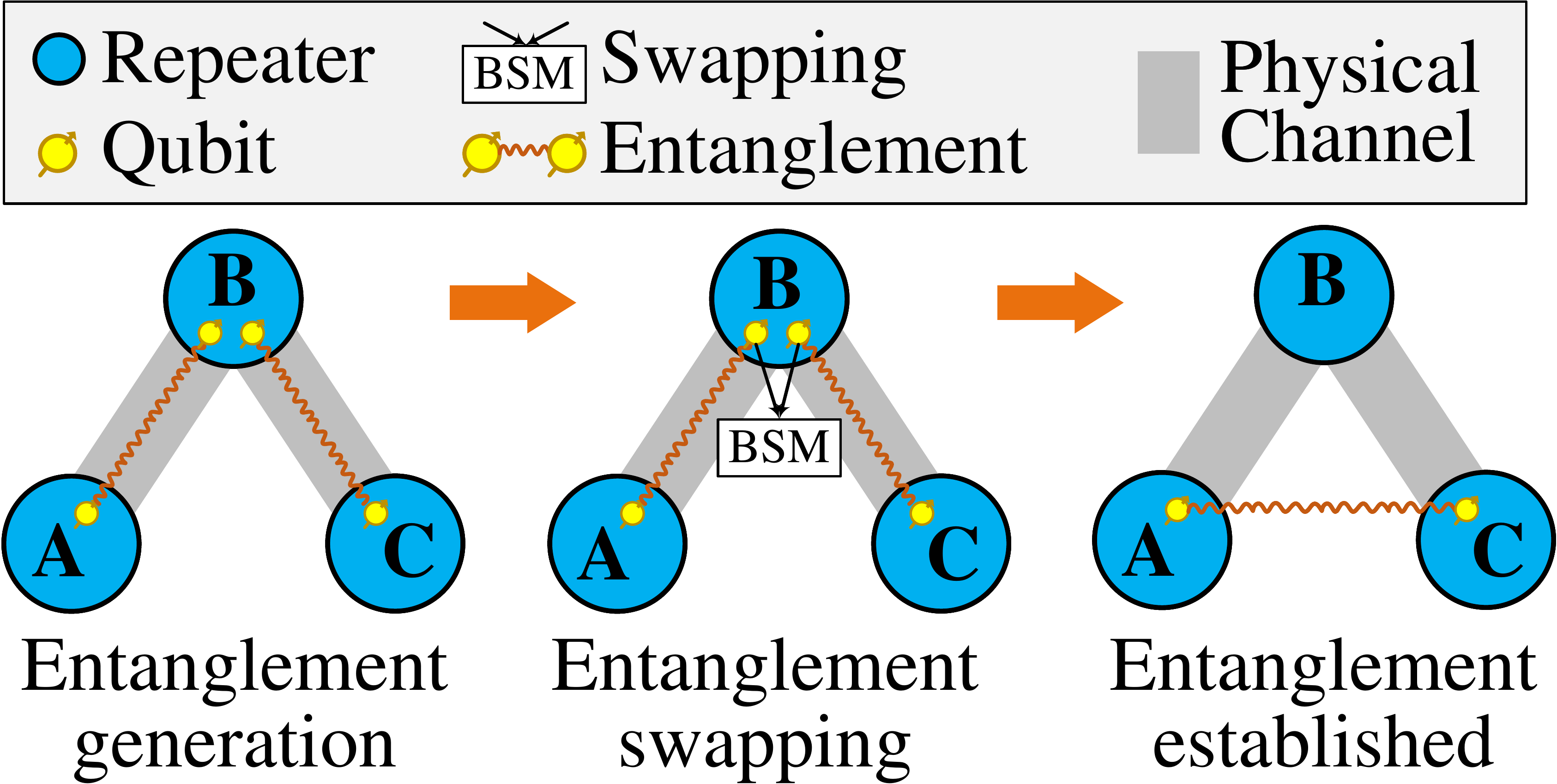}
\caption{\smallfont Basic quantum network operations: entanglement generation and entanglement swapping.}
\label{fig:swap}
\end{figure}
Two \emph{elementary ebits} are first generated along links $A$--$B$ and $B$--$C$ via entanglement generation.
To swap, $B$ entangles and measures its two local qubits and sends the result to either $A$ or $C$ via classical communications.
According to the result, $A$ or $C$ applies a unitary operation on its qubit. 
If all operations succeed, the two qubits at $A$ and $C$ are then entangled without interacting with each other.
This can be done recursively along a path until an end-to-end ebit between source and destination\footnotemark{} is established for quantum information exchange.
\footnotetext{Although entanglements are undirected, we use traditional network terms ``source'' and ``destination'' to denote an undirected pair of end nodes involved in quantum communications for simplicity.}

\subsection{Quantum Network Model}
\noindent
A quantum internet is a distributed facility
distributing remote ebits between source-destination ({{SD}}) pairs, via entanglement generation and swapping.
Formally, we model a quantum internet with an undirected graph $G = (N, L)$, where $N$ is the set of quantum repeaters, and $L$ is the set of physical channels (links) between repeaters.
Each link $l \in L$ has a capacity $c_l \in \mathbb{Z}^+$, denoting the number of channels that can be attempted for ebit generation along the link; $\mathbb{Z}^+$ denotes the positive integer set.
To model the aforementioned probabilistic processes, we further assume each link $l \in L$ has a success probability, $q_l$, denoting the probability of successfully generating one elementary ebit over one channel in unit time; each repeater $n \in N$ also has a swapping success probability $q_n$.
We adopt a time-slotted system model following~\cite{zhao2022e2e,zhao2021redundant}, while all our definitions and algorithms can be trivially extended to continuous-time asynchronous operations~\cite{zhao2023scheduling}.
In each time slot, the following phases are carried out in order:
\begin{enumerate}
\item\textbf{Entanglement generation:} For a pair of nodes $mn \!\in\!L$ with a direct link, they will attempt to generate elementary $mn$-ebits at a pre-defined rate\footnotemark{}\footnotetext{We use $mn$ to abbreviate an unordered node pair $\{ m,\! n \}$. Hence $mn \!=\! nm$.}.
\item\textbf{Entanglement swapping:} When ebits are available between both node pair $mk$ and node pair $kn$ sharing a common repeater $k$, repeater $k$ can attempt to perform entanglement swapping between each pair of $mk$- and $kn$-ebits to create ebits between remote nodes $m$ and $n$.

\end{enumerate}

We assume a central controller controls entanglement generation and swapping in the network~\cite{zhao2021redundant,shi2020concurrent,zhao2022e2e} by defining the rates of generation and directions of swapping across all links or nodes. 
The controller collects network-wide information, monitors network status, such as success probabilities and fidelities, and allocates resources in the network.

\vspace{-0.5em}
\subsection{Quantum Noise and Fidelity}

\noindent
While the above models assume perfect quantum channels and operations---meaning the final distributed ebits are the exact same state as the generated ones---the inevitable noise in quantum operation and transmission can introduce error and make the final state to differ from the initial state.
In classical communication, errors can be measured, detected and corrected on-the-fly or end-to-end.
In quantum, however, errors cannot be detected without destroying the quantum state due to the quantum no-cloning theorem.
Thus when a pure entangled state is affected by noise, it becomes a \emph{mixed state} that cannot be distinguished from the pure state without measurement.

Let $|\Phi^+\rangle$ be our desired pure entangled state\footnotemark{}.
\footnotetext{Since all Bell states are symmetric, we use $|\Phi^+\rangle$ as the desired state without loss of generality throughout this paper.}
A mixed state $M$ can result from $|\Phi^+\rangle$ going through a noisy channel, or noise in quantum operations.
Fidelity is a key quantum metric quantifying how close a mixed state is to the desired state, defined as $F \triangleq \langle \Phi^+| M | \Phi^+ \rangle$, and denoting the probability that $M$ (represented by a density matrix) is in the desired state $| \Phi^+ \rangle$.
To provide rigorous fidelity guarantee, we assume a worst-case isotropic error model~\cite{victora2020purification}, as compared to the bit flip error model in~\cite{zhao2022e2e}.
As shown by Bennett~\kETAL~\cite{bennett1996mixed}, an arbitrary mixed state $M$ with fidelity $F$ can be transformed to a Werner state with the same $F$ via \emph{random bilateral rotations} (RBR).
The Werner state with fidelity $F$ is defined as
\begin{equation*}
    \smallfont
    \mathcal{W}_F \!=\! F | \Phi^+ \rangle \langle \Phi^+ | + \frac{1\!-\!F}{3}( | \Phi^- \rangle \langle \Phi^- | + | \Psi^+ \rangle \langle \Psi^+ | + | \Psi^- \rangle \langle \Psi^- | ).
\end{equation*}
This Werner state can be viewed as a mixture of the pure state $|\Phi^+ \rangle$ with isotropic noise~\cite{victora2020purification}.
Hereafter, we assume all elementary and intermediate mixed-state ebits are transformed to the Werner state above before further operation.

For an elementary ebit established along a physical channel, its fidelity is decided by the quantum circuit that generates the entanglement, and the channel noise during transmission.
We define $F_l \in [0, 1]$ to model the fidelity of elementary ebits generated along each physical link $l$.

Given two ebits with fidelity $F_1$ and $F_2$ respectively, 
consider a \emph{perfect} entanglement swap performed between the two ebits implemented via BSM that consists of a CNOT gate followed by two single-qubit gates.
Since the ebits are mixed with noises, even a perfect entanglement swap may still fail due to the two ebits not being in the desired state $|\Phi^+\rangle$, leading to measurement error.
Two cases may result in a successful swap:
1) both ebits were in $|\Phi^+\rangle$ with probability $F^* = F_1F_2$, in which case the swapped ebit is also in $|\Phi^+\rangle$;
2) both ebits were not in $|\Phi^+\rangle$ but had equal states, with probability $F^{**} = 3 \frac{(1-F_1)}{3}\frac{(1-F_2)}{3}$, in which case the swapped ebit is in another Bell state instead of $| \Phi^+ \rangle$, but can be transformed to $| \Phi^+ \rangle$ via LOCC~\cite{bennett1996mixed}.
In the other cases, the swap fails because of unknown and unequal states of the two ebits.%

By combining these cases,
a \emph{perfect} entanglement swap will result in a new ebit with fidelity $F'$~\cite{Dur1999}, where
\begin{equation}
    \smallfont
    \label{eq:fed1}
    F' = F^* + F^{**} = \frac{1}{4} \cdot \left( 1 + 3\frac{(4F_1 - 1)}{3} \frac{(4F_2 - 1)}{3} \right).
\end{equation}

In practice, the swapping operation is also noisy or imperfect, and hence incurs additional fidelity loss.
Such loss is due to the (un)reliability of BSM, 1-qubit operation, and 2-qubit operation involved.
For instance, if a swap is performed with two elementary ebits with $F_1$ and $F_2$ at a node $n$ where the accuracy of BSM and probabilities of ideal 1-qubit, 2-qubit operations are $\alpha_n$, $o_{1,n}$, and $o_{2,n}$, respectively, the fidelity of a successfully generated ebit after swapping \cite{Dur1999} is 
\begin{equation}
\label{eq:newf}
    \smallfont
    F' = \frac{1}{4} \cdot \left( 1 + 3o_{1,n}o_{2,n}\frac{4\alpha_n^2-1}{3}\frac{4F_1 - 1}{3}\frac{4F_2 - 1}{3} \right).
\end{equation}

Based on Eq.~\eqref{eq:newf}, we facilitate notation by defining fidelity parameters $W_l \triangleq \frac{4F_l - 1}{3}$ and $W_n \triangleq o_{1,n}o_{2,n}\frac{4\alpha_n^2-1}{3}$ for each link $l$ and repeater $n$ respectively, and the fidelity of a successfully generated ebit after swapping is
\begin{equation}
    \label{eq:fed2}
    \smallfont
    F' = \frac{1}{4} \cdot \left( 1 + 3W_1W_2W_n \right).
\end{equation}
Assume an end-to-end ebit is established by swapping elementary ebits created along links $\{ l_1, l_2, \dots, l_{X+1} \}\! \subseteq L$ recursively at nodes $\{ n_1, n_2, \dots, n_{X} \} \!\subseteq\! N$.
Recursively applying Eq.~\eqref{eq:fed2}, the end-to-end fidelity of the ebit is
\begin{equation}
    \label{eq:fed3}
    \smallfont
    F^{\sf E2E} = \frac{1}{4} \cdot \left( 1 + 3 \prod_{i = 1}^{X+1} W_{l_i} \prod_{j = 1}^{X} W_{n_j} \right).
\end{equation}
From Eq.~\eqref{eq:fed3}, the end-to-end fidelity decreases exponentially with increasing number of hops~\cite{kozlowski2020designing}. 
Eq.~\eqref{eq:fed3} will serve as the basic tool to quantify and optimize the end-to-end fidelity of ebits distributed in a quantum internet.

Note that fidelity cannot be measured for a single ebit---the measurement itself will destroy the ebit.
As such, fidelity parameters can only be inferred from measuring and profiling ebits generated on elementary links (or after swap) for many times.
We assume that each node or link will be independently profiling the $W_n$ or $W_l$ value continuously throughout the network operation, and will regard these values (or some binary encoding of them) as input to further modeling and formulation.

\subsection{Network Performance Metrics}

\noindent
When optimizing operation of a quantum network, two performance metrics have been widely considered in the literature.

\noindent\textbf{Entanglement distribution rate (EDR):} 
similar to throughput in classical network, EDR is the number of ebits distributed between an SD pair in unit time. 
Due to the probabilistic operations, we use $\eta_{st}$ to denote the \emph{expected EDR} between the source $s$ and destination $t$.

\noindent\textbf{End-to-end fidelity:}
as another desired metric, a higher end-to-end fidelity $F^{\sf E2E}$ leads to higher communication efficiency.

\section{Characterizing Achievable EDR and Fidelity} 
\label{sec:abs}

\subsection{Characterizing Achievable Expected EDR}
\noindent
We start with the question of how to characterize the maximum achievable EDR between two nodes in a given network.
Assuming no quantum memory is available, the generated ebits would decohere within one time slot, meaning that all generation and swapping processes along an end-to-end path must succeed within one time slot in order to successfully generate an end-to-end ebit.
The probability of successful generation along one path is thus the product of all node and link probabilities: $P^{\sf E2E}_{st} = \prod_{i = 1}^{X+1} q_{l_i} \prod_{j = 1}^{X} q_{n_j}$ for a path $\rho = (n_0, n_2, \dots, n_{X+1})$ where $s = n_0$, $t = n_{X+1}$ and $l_i = n_{i-1}n_i \in L$.
The achievable expected EDR is then the bottleneck capacity $c^*_{st} \triangleq \min_{i} \{ c_{l_i} \}$ times the end-to-end success probability: $\eta_{st} = c^*_{st} \cdot P^{\sf E2E}_{st}$.

It is expected that future quantum repeaters will be equipped with quantum memories acting as temporary buffers. 
In this case, rate characterization becomes more complicated.
In~\cite{shi2020concurrent}, it has been shown that post-selection and storage can increase the maximum achievable EDR beyond the simple product of probabilities times capacity, since the quantum memories can temporarily buffer and rematch the post-selected ebits that are unmatched for swapping due to unsuccessful ebit generation on other links.
Subsequently, many works have explored how to design entanglement routing and distribution protocols with limited or ephemeral quantum memories to improve EDR~\cite{zhao2022e2e,zhao2021redundant,zhao2023scheduling}.
However, to what extent can post-selection and storage increase the optimal expected EDR remains unclear.

A recent breakthrough is a tight \emph{upper bound} on the maximum achievable EDR between a pair of nodes with post-selection and storage, due to Dai~\kETAL~\cite{Dai2020a, Dai2020b}.
Their result is based on an abstraction called the \emph{entanglement flow}, or \textbf{eflow}, which formulates the maximum achievable expected EDR as a linear program.
Below, we present the definition of an eflow in~\cite{Dai2020b}, slightly modified to align with our notation, which will be used subsequently in our decomposition theorem.

\begin{definition}[Eflow~\cite{Dai2020b}]
\label{def:eflow}
Given a network $G = (N, L)$ and an SD pair $st$, an eflow in $G$ is defined by variables
\begin{itemize}[noitemsep,topsep=0pt]
    \item $g_{m{{{{{}}}}}n} \in [0, 1], \forall mn \in L$, denoting the rate of elementary ebit generation along the physical link $mn$, as a ratio of the capacity $c_{mn}$ of the link, and
    \item $f^{m{{{{{}}}}}k}_{m{{{{{}}}}}n} \ge 0, \forall m, n, k \in N$, denoting the expected rate of ebits established between nodes $m$ and $k$ that will be used for swapping to generate ebits between nodes $m$ and $n$.
\end{itemize}
A feasible eflow must have $f$ and $g$ satisfying:
\begin{subequations}
\label{fml:ored}
\begin{align}
    & f^{m{{{{{}}}}}k}_{m{{{{{}}}}}n} = f^{k{{{{{}}}}}n}_{m{{{{{}}}}}n}, 
    \quad \forall m, n, k \in N;    
    \label{eq:eflow:constr_a} 
    \\
    & I(mn) \!= \!\Omega(mn), 
    \quad \forall m, n \!\in\! N, mn \ne st; 
    \label{eq:eflow:constr_b}
    \\
    & \Omega(st) = 0; 
    \label{eq:eflow:constr_c}
\end{align}
where for $\forall m, n \in N$,
\begin{align}
     \label{eq:I}
    &I(mn) \triangleq 
        q_{mn} c_{mn} g_{m{{{{{}}}}}n} \!\cdot\!\mathbf{1}_{mn \in L} \!+\!\! 
        \sum_{\mathclap{k \in N \setminus \{m, n\}}}  
        \frac{q_k}{2} \! \left( f^{m{{{{{}}}}}k}_{m{{{{{}}}}}n} \!+\! f^{k{{{{{}}}}}n}_{m{{{{{}}}}}n} \right) \!\! ,   
\\
    \label{eq:O}
    &\Omega(mn) \triangleq 
         \sum_{\mathclap{k \in N \setminus \{m, n\}}}  
            \left( f^{m{{{{{}}}}}n}_{m{{{{{}}}}}k} + f^{m{{{{{}}}}}n}_{k{{{{{}}}}}n} \right),
\end{align}
\end{subequations}
and $\mathbf{1}_{mn \in L}$ is an indicator function of whether $mn \in L$.
The \textbf{eflow value} of SD pair $st$ is defined as $\eta_{st} \triangleq I(st)$.
\myendbox
\end{definition}

\textbf{\emph{Explanation:}} 
For brevity, each node pair $m{{{{{}}}}}n$ with $m, n \in N$ is called an \textbf{enode} hereafter, meaning that post-selected ebits may be established between the pair of nodes at some stage of remote distribution.
Here $I(mn)$ denotes the ebits generated between node pair $m{{{{{}}}}}n$ (including elementary ebits and ebits generated by swapping), and $\Omega(mn)$ denotes the ebits contributed by $m{{{{{}}}}}n$ to generating ebits between other node pairs via swaps.
Note that the elementary ebits (first term in Eq.~\eqref{eq:I}) are discounted by generation probability $q_{mn}$ of a link, and ebits received from swapping at node $k$ are discounted by node $k$'s swapping probability $q_k$.
Eq.~\eqref{eq:eflow:constr_a} enforces the two node pairs $m{{{{{}}}}}k$ and $k{{{{{}}}}}n$, whose ebits will be swapped to form ebits for $m{{{{{}}}}}n$, contribute equal number of ebits to the swap.
Eq.~\eqref{eq:eflow:constr_b} enforces an intermediate pair $m{{{{{}}}}}n$ does not keep generated ebits, but contributes all ebits to further swapping for establishing end-to-end $st$-ebits.
Eq.~\eqref{eq:eflow:constr_c} constrains that the SD pair $st$ should not contribute any established ebits to further swapping.
An eflow describes how ebits ``flow through'' different enodes and ``merge'' at repeaters until some are ``landed in'' (established between) the SD pair's enode $st$, with flow conservation at repeaters enforced by~\eqref{eq:eflow:constr_b}.

One way to visualize an eflow is to define its \emph{induced graph}, $\mathcal{G} = (\mathcal{V}, \mathcal{E})$, where $\mathcal{V} \subseteq (N \times N) \cup \{ \bot \}$ is a set of enodes (with a special enode $\bot$ denoting the generation process), and $\mathcal{E} \subseteq \mathcal{V} \times \mathcal{V}$ are directed edges denoting generation and swapping processes.
An enode $mn \in \mathcal{V}$ corresponds to one with $I(mn) > 0$ in the eflow.
An edge $(mk, mn) \in \mathcal{E}$ then denotes one swapping variable $f^{mk}_{mn} > 0$.
An edge $(\bot, mn) \in \mathcal{E}$ specially denotes a generation variable $g_{mn} > 0$.
From Eq.~\eqref{eq:eflow:constr_a}, it is clear that swapping edges $(mk, mn)$ and $(kn, mn)$ must appear simultaneously in $\mathcal{G}$---either they both present or they both absent.
An example is shown in Fig.~\ref{fig:induced-graph}.

\begin{figure}[t!]
\centering
\includegraphics[width=0.5\textwidth]{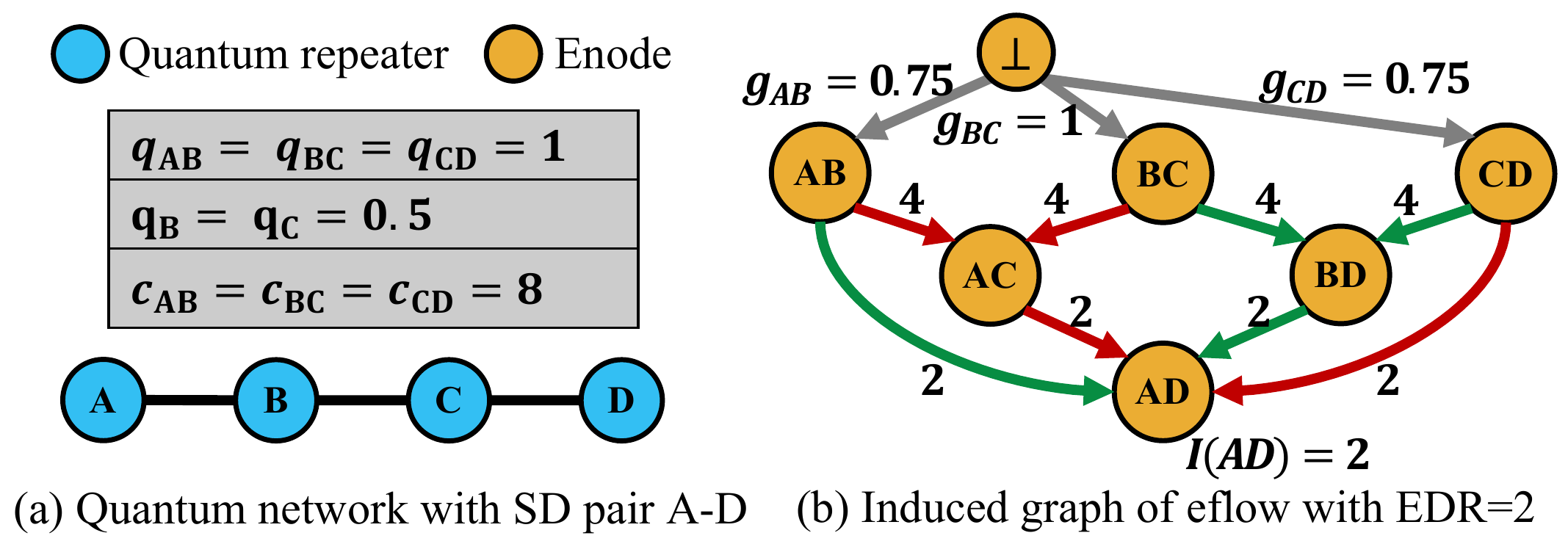}
\caption{\smallfont Induced graph of an eflow. Value on each edge denotes a variable: $g_{mn}$ or $f^{mk}_{mn}$. Two same-color edges pointing to one enode are matched for swapping.}
\label{fig:induced-graph}
\end{figure}

We summarize the importance of the eflow formulation with the following theorem, which restates Theorems~1--2 in~\cite{Dai2020b}.
\begin{theorem}[Characterizing maximum EDR~\cite{Dai2020b}]
    \label{th:12}
    The optimal solution to the following problem (called ORED in~\cite{Dai2020b}),
    \begin{equation}\label{fml:ored_new}\eta^*_{st} \!\triangleq\! \max\nolimits_{f, g} \{ \eta_{st} \,|\, \eqref{eq:eflow:constr_a}\text{--}\eqref{eq:O} \},\end{equation} 
    is a tight upper bound on the maximum expected EDR between $s$ and $t$ in $G$. 
    The induced graph $\mathcal{G}$ of at least one optimal solution is a directed acyclic graph (DAG).
    Furthermore, there exists an entanglement distribution protocol that can achieve expected EDR of $\eta^*_{st}$ between $st$.
    \myendbox
\end{theorem}
A stochastic protocol achieving $\eta^*_{st}$ based on post-selection and queueing was proposed in~\cite{Dai2020b}, which we shall extend in the Appendix to achieve both high EDR and fidelity.

The primary limitation with the eflow formulation is that it \textbf{cannot model ebit fidelity loss} in generation and swapping.
Since ebits may arrive at an enode from any possible sequence of swaps at arbitrary repeaters, there may be an exponential number of possible paths in $G$ from which an ebit might have been generated, and some may result in low fidelity that can render the distributed ebits unusable.
In the next subsection, we propose a novel abstraction, called \emph{primitive eflow}, to characterize the end-to-end fidelity of the distributed ebits.

\subsection{Eflow Decomposition \& Characterizing End-to-end Fidelity}
\noindent
The limitation with eflow is that there is no way of tracking along which path is an arbitrary ebit generated, as many paths may have been utilized to generate end-to-end ebits defined by an eflow, possibly exponentially many.
Thus to characterize the end-to-end fidelity, we need a new abstraction that can naturally encode the fidelity of entanglement paths while still leading to the same tight upper bound on achievable EDR.
The path formulation itself cannot fulfill the second part of the goal (characterizing EDR) since EDR-optimal entanglement routing (with post-selection) is sill an open problem.

In the following, we find that a special type of eflow serves to combine the two goals (EDR and fidelity characterization).
This abstraction, named \textbf{primitive eflow (pflow)}, enables an alternative formulation that is equivalent to Program~\eqref{fml:ored}, similar to the path-flow formulation in classical network flow as an alternative to the edge-flow formulation~\cite{Ahuja1993}.
We will establish this equivalence with a novel \emph{eflow decomposition theorem}.

\begin{definition}[Pflow]
\label{def:pflow}
A \emph{primitive eflow (pflow)} is a feasible eflow as defined by Program~\eqref{fml:ored}, which additionally satisfies that: for every enode $m{{{{{}}}}}n$, either $g_{m{{{{{}}}}}n} \!>\! 0$, or there exists exactly one $k \!\in\! N$ such that $f^{m{{{{{}}}}}k}_{m{{{{{}}}}}n} \!=\! f^{k{{{{{}}}}}n}_{m{{{{{}}}}}n} \!>\! 0$, but not both.
\myendbox
\end{definition}

A pflow is \emph{primitive} in that ebits at each enode $m{{{{{}}}}}n$ is generated in exactly one way: 
either they are elementary ebits generated directly along link $mn \in L$, or they are generated by swapping $m{{{{{}}}}}k$- and $k{{{{{}}}}}n$-ebits at a single intermediary $k$.
The induced graph $\mathcal{G}$ of a pflow, excluding the special $\bot$ vertex, is always a binary tree rooted at enode $st$ by the definition; Fig.~\ref{fig:induced-graph} shows two such binary trees with different colors.
A pflow naturally represents exactly one \emph{path} in the quantum internet, and the final ${{{{{st}}}}}$-ebits generated along a pflow have \textbf{identical fidelity}, which can be directly computed via Eq.~\eqref{eq:fed3}.

Another property of a pflow is that the ratio between each variable in $\{ g_{m{{{{{}}}}}n}, f^{m{{{{{}}}}}k}_{m{{{{{}}}}}n} \,|\, m, k, n \in N \}$, and the end-to-end EDR $\eta_{st}$, is fixed.
Let $\bar g_{m{{{{{}}}}}n}$ or $\bar f^{m{{{{{}}}}}k}_{m{{{{{}}}}}n}$ be the ratio between the corresponding variable and the EDR of the pflow.
Given the induced graph $\mathcal{G}$ of the pflow, these ratios can be computed as in Algorithm~\ref{a:pcost},
backtracking from enode $s{{{{{}}}}}t$ which has a ratio of $1$ (one generated ebit between $s{{{{{}}}}}t$ translates into one end-to-end ${{{{{st}}}}}$-ebit).
For each enode $m{{{{{}}}}}n$, its output ebit rate $\Omega(mn)$ is added to its input ebit rate(s), \kIE, either $\bar g_{m{{{{{}}}}}n}$ or $\bar f^{m{{{{{}}}}}k}_{m{{{{{}}}}}n}$ and $\bar f^{k{{{{{}}}}}n}_{m{{{{{}}}}}n}$ for some $k$, augmented by the corresponding expected ratios of $1/q_{mn}$ or $1/q_k$ respectively.
Based on Algorithm~\ref{a:pcost}, a pflow can essentially be defined by its induced graph $\mathcal{G}$, and a single objective value $\eta_{\mathcal{G}}$ assigned to this pflow.

\begin{algorithm}[t]
\smallfont
\caption{\mbox{Computing ebit generation ratios of a pflow}}
\label{a:pcost}
\KwIn{Induced graph $\mathcal{G}$ of an ${{{{{st}}}}}$-pflow}
\KwOut{Ebit generation ratios $\{ \bar g_{m{{{{{}}}}}n}, \bar f^{m{{{{{}}}}}k}_{m{{{{{}}}}}n}, \bar f^{k{{{{{}}}}}n}_{m{{{{{}}}}}n}\}$}
Initialize all ratios to $0$, and $Q \leftarrow \{ ( s{{{{{}}}}}t, 1) \}$\;
\While{$Q \ne \emptyset$}{
    $(m{{{{{}}}}}n, \psi) \leftarrow Q.pop()$\;
    \eIf{$\nexists k$ such that $(m{{{{{}}}}}k, m{{{{{}}}}}n) \in \mathcal{E}$}{
        $\bar g_{m{{{{{}}}}}n} \leftarrow g_{m{{{{{}}}}}n} + \psi / (q_{mn} \cdot c_{mn})$\;
    }{
        $\bar f^{m{{{{{}}}}}k}_{m{{{{{}}}}}n} \leftarrow \bar f^{m{{{{{}}}}}k}_{m{{{{{}}}}}n} + \psi / q_k$, $\bar f^{k{{{{{}}}}}n}_{m{{{{{}}}}}n} \leftarrow \bar f^{k{{{{{}}}}}n}_{m{{{{{}}}}}n} + \psi / q_k$\;
        $Q.push( ( m{{{{{}}}}}k, \psi / q_{k} ) )$, $Q.push( ( k{{{{{}}}}}n, \psi / q_{k} ) )$\;
    }
}
\Return{$\{ \bar g_{m{{{{{}}}}}n}, \bar f^{m{{{{{}}}}}k}_{m{{{{{}}}}}n}, \bar f^{k{{{{{}}}}}n}_{m{{{{{}}}}}n}\}$.}
\end{algorithm}

Crucially, the pflow abstraction leads to the following theorem, which generalizes the classical \emph{flow decomposition theorem}~\cite{Ahuja1993} to the quantum network setting:

\begin{theorem}[Eflow decomposition]
\label{th:decomp}
An eflow with $\eta_{st} > 0$ can be decomposed into a polynomial number of pflows.
\myendbox
\end{theorem}
\begin{proof}
Let $\mathcal{G}$ be the induced graph of the eflow.
We first find an induced graph $\mathcal{G}' \subseteq \mathcal{G}$ in which each enode $m{{{{{}}}}}n \in \mathcal{G}'$ has either $g_{m{{{{{}}}}}n} > 0$, or there is exactly one $k \in N$ such that $f^{m{{{{{}}}}}k}_{m{{{{{}}}}}n} = f^{k{{{{{}}}}}n}_{m{{{{{}}}}}n} > 0$, and $(m{{{{{}}}}}k, m{{{{{}}}}}n)$ and $(k{{{{{}}}}}n, m{{{{{}}}}}n)$ are both in $\mathcal{G}'$.
Such a subgraph must exist due to the constraint of $I(mn) - \Omega(mn) = 0$ for every $m{{{{{}}}}}n \ne s{{{{{}}}}}t$, and that $\eta_{st} > 0$ for the eflow.
We then use Algorithm~\ref{a:pcost} to compute ebit generation ratios of the pflow corresponding to $\mathcal{G}'$.
Let $\eta^*$ be the maximally acceptable EDR of this pflow. We calculate it as $\eta^* \triangleq \min ( \{ f^{m{{{{{}}}}}k}_{m{{{{{}}}}}n} / \bar f^{m{{{{{}}}}}k}_{m{{{{{}}}}}n} \,|\, m, n, k \in N, \bar f^{m{{{{{}}}}}k}_{m{{{{{}}}}}n} > 0 \} \cup \{ g_{m{{{{{}}}}}n} / \bar g_{m{{{{{}}}}}n} \,|\, m, n \in N, \bar g_{m{{{{{}}}}}n} > 0 \} )$.
Assigning $\eta^*$ to this pflow, we can update the original eflow by deducting each variable by $\eta^*$ times the corresponding ebit generation ratio in the pflow.
Continue this process until $\eta_{st} = 0$, and we arrive at a set of pflows with sum of EDRs equal to $\eta_{st}$.

In the above process, either at least one $g_{m{{{{{}}}}}n}$, or at least one pair of $\{ f^{m{{{{{}}}}}k}_{m{{{{{}}}}}n}, f^{k{{{{{}}}}}n}_{m{{{{{}}}}}n} \}$ variables with some $k$, becomes $0$ after updating each pflow.
Since there are in total $O(N^3)$ variables, this decomposition results in at most $O(N^3)$ pflows. 
\end{proof}%
\vspace{-1em}
Theorem~\ref{th:decomp} enables an alternative pflow-based formulation to Program~\eqref{fml:ored} in Definition~\ref{def:eflow}.
Let $\Psi_{st}$ be the set of all possible pflows between $s$ and $t$, and let $\eta(\psi) \ge 0$ be the pflow value assigned to $\psi \in \Psi_{st}$.
Lemma~\ref{l:alt_formulation} follows from Theorem~\ref{th:decomp}:
\begin{lemma}[Pflow-based EDR Characterization]
    \label{l:alt_formulation}
    The maximum expected EDR $\eta^*_{st}$ in Eq.~\eqref{fml:ored_new} can be computed by Program~\eqref{eq:alt_form}:
    \begin{equation}
        \label{eq:alt_form}
        \begin{aligned}
            \eta^*_{st} = \max\nolimits_{\eta} & \;\;
            \sum\nolimits_{\psi \in \Psi_{st}} \eta(\psi)  
            \\
            \text{ \emph{s.t.} } & \;\;
            \sum_{\mathclap{\psi \in \Psi_{st}: mn \in \psi}} {\bar g}_{mn} \cdot \eta(\psi) \le 1,
            \quad \forall mn \in L.
        \end{aligned}
    \end{equation}
\end{lemma}

Program~\eqref{eq:alt_form} computes $\eta^*_{st}$ by assigning values to pflows in $\Psi_{st}$, while making sure that no link $mn \in L$ is oversubscribed with a ratio greater than $1$, i.e., being asked to generate more than $c_{mn}$ ebits per unit time.
From this formulation, the key observation is that, each actual ebit is still generated along a single entanglement path.
The fidelity of the ebit is precisely defined by the underlying path along which it is generated based on Eq.~\eqref{eq:fed3}.
Assume an eflow is able to generate ebits all with fidelity no less than a given bound $\Upsilon_{st}$.
Following Lemma~\ref{l:alt_formulation}, the eflow can always be decomposed into a set of pflows, where each pflow generates ebits along a fixed path with fidelity lower bounded by $\Upsilon_{st}$ (some of the pflows may share the same path).
This leads to Theorem~\ref{th:fid}.

\begin{theorem}[Characterizing worst-case fidelity]
\label{th:fid}
An eflow that generates ebits with minimum end-to-end fidelity of $\Upsilon_{st}$ can be decomposed into a set of pflows, each along an ${{{{{st}}}}}$-path whose fidelity is at least $\Upsilon_{st}$.
\myendbox
\end{theorem}
\textbf{\emph{Remark:}}
The importance of Theorem~\ref{th:decomp} is not to characterize the maximum end-to-end fidelity between $st$ for generating a single ebit.
Such maximum fidelity can be easily computed by employing a Dijkstra's algorithm and finding a highest-fidelity path following Eq.~\eqref{eq:fed3}.
Instead, the goal is to characterize the worst-case end-to-end fidelity for achieving an end-to-end EDR goal, or vice versa, utilizing as many paths/pflows as possible. 
Next, we motivate and then formally define the problem of characterizing the EDR-fidelity trade-off in quantum network.
%
%
%



\subsection{Trade-off Between EDR and Worst-case Fidelity}
\noindent
Consider a quantum application having two performance requirements for remote entanglement distribution: 1) the long-term average EDR is at least $\Delta_{st}$; 2) each generated ebit has fidelity no less than $\Upsilon_{st}$.
Having a higher EDR goal $\Delta_{st}$ means the network may need to utilize more paths for distribution, some maybe leading to lower end-to-end fidelity than others, which overall may lead to a lower $\Upsilon_{st}$ that can be satisfied.

\begin{figure}[t]
\centering
\subfloat[A network with three quantum links and different fidelity. All links have unit capacity and $1.0$ success probability.]{\includegraphics[width=0.185\textwidth,valign=t]{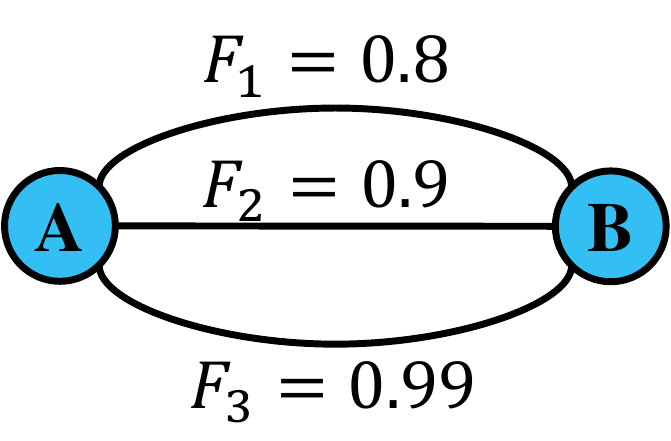}
\label{fig:pareto:example}}
\hfil
\subfloat[EDR-fidelity trade-off curve for SD pair $AB$, and a solution in the frontier.]{\includegraphics[width=0.25\textwidth,valign=t]{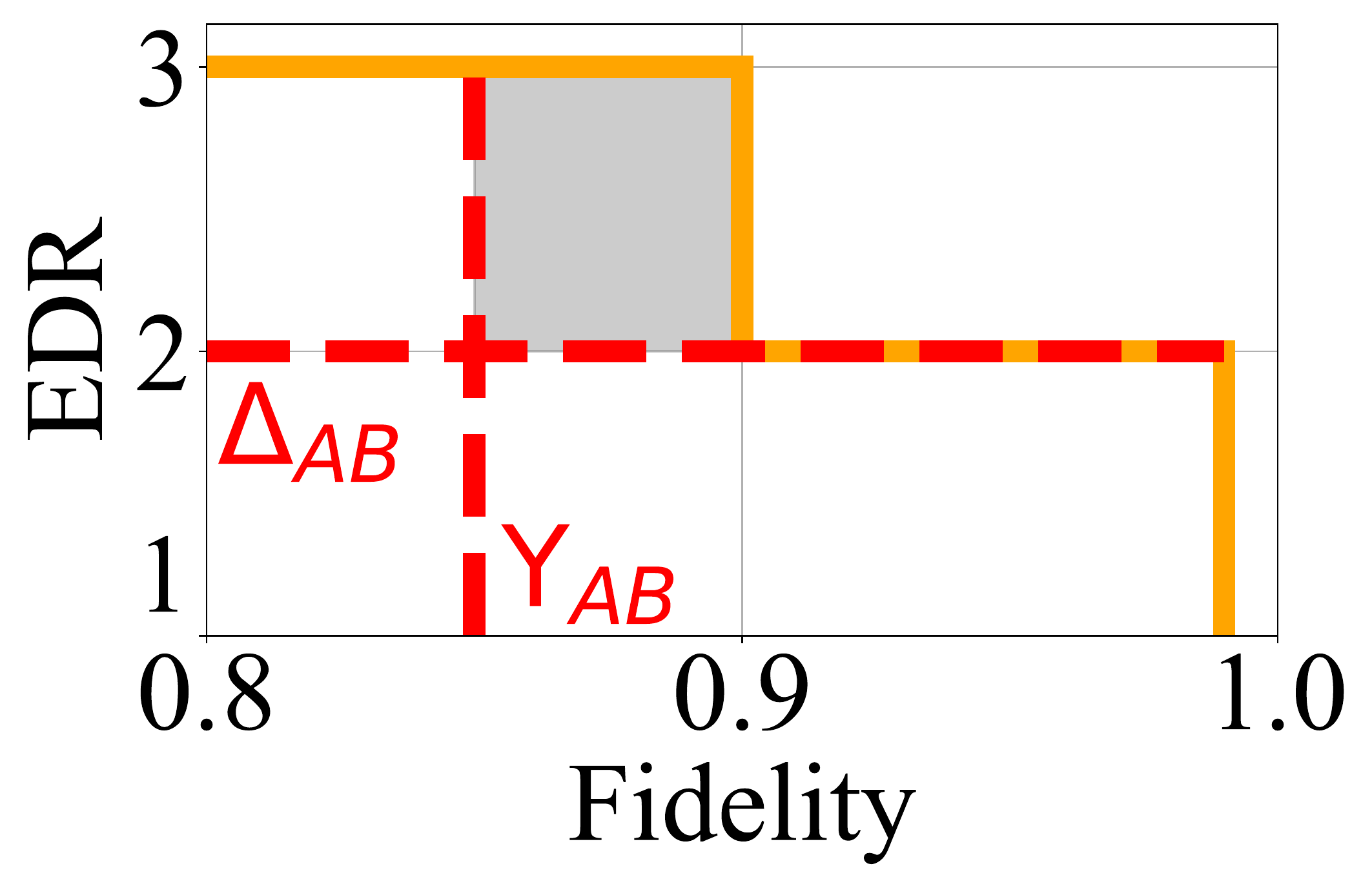}
\label{fig:pareto:hf_red}}
\hfil
\caption{The EDR-fidelity Pareto frontier of a simple network. Shaded area denotes gap from a solution to the actual frontier.}
\label{fig:pareto}
\end{figure}

Fig.~\ref{fig:pareto} shows a simple example to motivate the EDR-fidelity trade-off.
Consider an SD pair $A$ and $B$ in Fig.~\ref{fig:pareto}\subref{fig:pareto:example}, which are connected by three different quantum links, all with capacity $1$ but different fidelity.
When the end-to-end fidelity requirement increases, the achievable EDR will decrease as the number of feasible paths/pflows becomes less, and vice versa, as shown in Fig.~\ref{fig:pareto}\subref{fig:pareto:hf_red}.
The trade-off can become more complicated when swapping probability and fidelity loss are taken into account. 
We start to explore this trade-off from the above motivating example, where the objective is to simultaneously satisfy the expected EDR and fidelity goals of an application, as shown in Fig.~\ref{fig:pareto}\subref{fig:pareto:hf_red}.
To this end, we define the high-fidelity remote entanglement distribution (HF-RED) problem.
\begin{definition}[HF-RED]
\noindent
Given a quantum network $G = (N, L)$ and an SD pair $st$,
let $\Delta_{st} \!> \!0$ be the expected EDR bound %
and $\Upsilon_{st}\! >\! 0$ be the end-to-end fidelity bound. The \textbf{high-fidelity remote entanglement distribution} problem (denoted as \textbf{HF-RED}) is to seek {a set of pflow} $\mathcal{P}_{st}^* \subseteq \Psi_{st}$, which delivers end-to-end $st$-ebits 
satisfying that \\
1) total expected EDR $\eta_{st}$ of all pflows is at least $\Delta_{st}$, and \\
2) each pflow has fidelity no less than $\Upsilon_{st}$. 
\myendbox
\end{definition}
Without loss of generality, we further define an optimization version of HF-RED, which we call the \textbf{OF-RED} problem, for maximizing the worst-case end-to-end fidelity subject to the expected EDR bound.
We note that OF-RED is an important problem for characterizing the EDR-fidelity trade-off.
Particularly, one can apply the the well-known $\epsilon$-constraint method in multi-objective optimization~\cite{Miettinen1998} to find \emph{weak Pareto optimal solutions}---solutions that cannot be improved on one of the metrics without hurting the other---by repetitively solving OF-RED with different bounds on the expected EDR.
In Sec.~\ref{pe}, we will utilize this method to characterize the EDR-fidelity trade-off curve of a given quantum network and SD pair.
This approach depends on solving OF-RED efficiently, which, nevertheless, is highly non-trivial as we will show next.

\subsection{Computational Complexity}
\label{sec:complexity}

\noindent
Let $\mathcal{P}_{st} \subseteq \Psi_{st}$ be the set of $st$-pflows that are along paths with fidelity no lower than $\Upsilon_{st}$.
HF-RED can be easily formulated based on Program~\eqref{eq:alt_form}, by replacing $\Psi_{st}$ with $\mathcal{P}_{st}$ in the formulation---this constrains the program to only use pflows satisfying the end-to-end fidelity constraint $\Upsilon_{st}$ when trying to achieve the EDR goal $\Delta_{st}$.
Notably, both Program~\eqref{eq:alt_form} and this fidelity-aware version are linear programs (LPs), but with exponential sizes due to the potentially exponential number of possible pflows in $\Psi_{st}$ (or $\mathcal{P}_{st}$).
In fact, the following lemma demonstrates the computational complexity of this problem:

\begin{lemma}
\label{l:np}
HF-RED and OF-RED are NP-hard.
\myendbox
\end{lemma}
\begin{proof}
We prove NP-hardness of HF-RED by a reduction from the \emph{Multi-Path routing with Bandwidth and Delay constraints (MPBD)} problem, which is NP-complete~\cite{misra2009polynomial}.
Given a graph, an SD pair and two values $B, D > 0$, MPBD asks for a set of paths with delay upper bounded by $D$, and a network flow over these paths with total flow lower bounded by $B$.
Given an MPBD instance, let us build an instance of HF-RED.
First, we set all probabilities $q_l$ and $q_n$ to $1$.
Then we set $W_l = e^{-d_l}$ where $d_l > 0$ is the delay of link $l$, and $W_n = 1$ for $n \in N$.
Note that since $d_l > 0$, $W_l \in (0, 1)$.
The fidelity bound is $\Upsilon_{st} = (1 + 3 \cdot e^{-D})/4$.
Capacity $c_l$ is set as the bandwidth in MPBD, and EDR bound $\Delta_{st} = B$.
Given this construction, any generated ebit represents a path $p$ such that $\prod_{l \in p} W_l = e^{-\sum_l {d_l}} \ge e^{-D}$, which gives $\sum_l d_l \le D$.
Meanwhile, for any delay-feasible path in MPBD, generating end-to-end ebits along this path will satisfy the fidelity bound $\Upsilon_{st}$.
Since generation and swapping both have success probability $1$, the EDR is exactly equal to the end-to-end $st$-flow value.
Hence a solution to MPBD gives a feasible solution to HF-RED, and vice versa.
HF-RED is thus NP-hard, and the NP-hardness of OF-RED follows.
\end{proof}

\textbf{\emph{Remark (from fidelity to length):}}
We utilize the above proof to transform end-to-end fidelity in Eq.~\eqref{eq:fed3} into an additive metric.
Define \emph{length} values $\zeta_l = -\log(W_l)$ and $\zeta_n = -\log(W_n)$ for link and node fidelity values, respectively.
Consider end-to-end fidelity $F^{\sf E2E}$ of a path in Eq.~\eqref{eq:fed3}. 
Define the path length as $Z = \sum_{i = 1}^{X+1} \zeta_{l_i} \sum_{j = 1}^{X} \zeta_{n_j}$, then $F^{\sf E2E} = \frac{1}{4} \cdot \left( 1 + 3 e^{-Z} \right)$.
Since the above transformation is bijective, maximizing the worst-case fidelity is equivalent to minimizing the longest path length.
Given a fidelity bound $\Upsilon_{st}$, it is also easy to define an equivalent length bound $\mathbf{Z}_{st} = -\log\left( \frac{4\Upsilon_{st}-1}{3} \right)$, such that any path with length upper bounded by $\mathbf{Z}_{st}$ will have fidelity lower bounded by $\Upsilon_{st}$, and vice versa.
Note that using either $W_l, W_n$ or $\zeta_l, \zeta_n$ only differs in the \emph{binary encoding} to represent the fidelity parameters.
Because of the equivalence, we next focus on \textbf{minimizing the maximum path length} in OF-RED. 
%
%



%

\begin{table}
\caption{Key Notations for Algorithm Design}
\label{notation2}
\footnotesize
\begin{tabular}{p{1.47cm}p{6.63cm}}
\hline
Parameters & Description\\
\hline
$m{{{{{}}}}}n/z$ & extended enode $m{{{{{}}}}}n$ with a path segment length of $z$\\
$\upsilon_{st}^*$ & optimal worst-case end-to-end fidelity between $s$ and $t$\\
{$\zeta_l, \zeta_n$} & length values of link $l$ and node $n$\\
$\varsigma$ & Boolean output of the approximate testing algorithm~\ref{a:test}\\
$\varepsilon$ &  approximation accuracy parameter\\
$\theta$ & quantization factor of node/link lengths\\
$\zeta(p), \zeta^{\theta}(p)$ & lengths of path $p$ before \& after quantization with $\theta$ \\
$\mathbf{Z}, Z$ & path length bounds before and after quantization\\
$\mathbf{Z}^*, Z^*$ & original optimal longest path length, and quantized value\\
$\mathbf{Z}^{\theta}$ & optimal longest path length for quantized OF-RED\\
$\text{LB}, \text{UB}$ & lower \& upper bounds on optimal longest path length $\mathbf{Z}^*$\\
$Z_{\text{LB}}, Z_{\text{UB}}$ & the quantized lower and upper bounds\\
\hline
\end{tabular}

\end{table}

\section{FPTAS for Optimizing Fidelity under EDR Bound}
\label{sec:fptas}
\noindent
The OF-RED problem aims to search for the highest worst-case fidelity $\upsilon_{st}^*$ (equivalently the minimum longest path length $\mathbf{Z}_{st}^*$) under a minimum end-to-end EDR requirement $\Delta_{st}$.
Directly solving the OF-RED problem can be NP-hard, which precludes us from designing efficient optimal algorithms for the problem.
Instead, we seek to design an approximation algorithm to OF-RED, which can then be used to characterize the approximate weak Pareto frontier of EDR-fidelity trade-off.
Our fully polynomial-time approximation scheme (FPTAS) for the OF-RED problem consists of four building blocks. 
Notations related to algorithms are summarized in Table~\ref{notation2}.

First, we design a pseudo-polynomial-time \emph{Fidelity-aware Optimal Remote Entanglement Distribution (FORED)} program as an extension to Program~\eqref{fml:ored}.
Under restrictive integrality conditions on the length values, the program outputs an eflow achieving maximum EDR with lower-bounded length (fidelity).

Our second building block, an \emph{approximate testing algorithm}, uses the FORED program as a sub-routine to test if a specific length value can be achieved with the EDR bound satisfied, subject to a small and bounded testing error.

Our third building block is a polynomial-time \emph{sorting and trimming algorithm}, which finds a pair of close-enough lower and upper bounds for the optimal length value, to serve as the initial range in which the optimal value will be searched for.

Finally, a \emph{two-stage bisection search algorithm} is devised to iteratively narrow down the initial range until a solution is found within a small approximation error of the optimal length (fidelity) value while satisfying the EDR bound.
\begin{figure}[t!]
\centering
\includegraphics[width=0.485\textwidth]{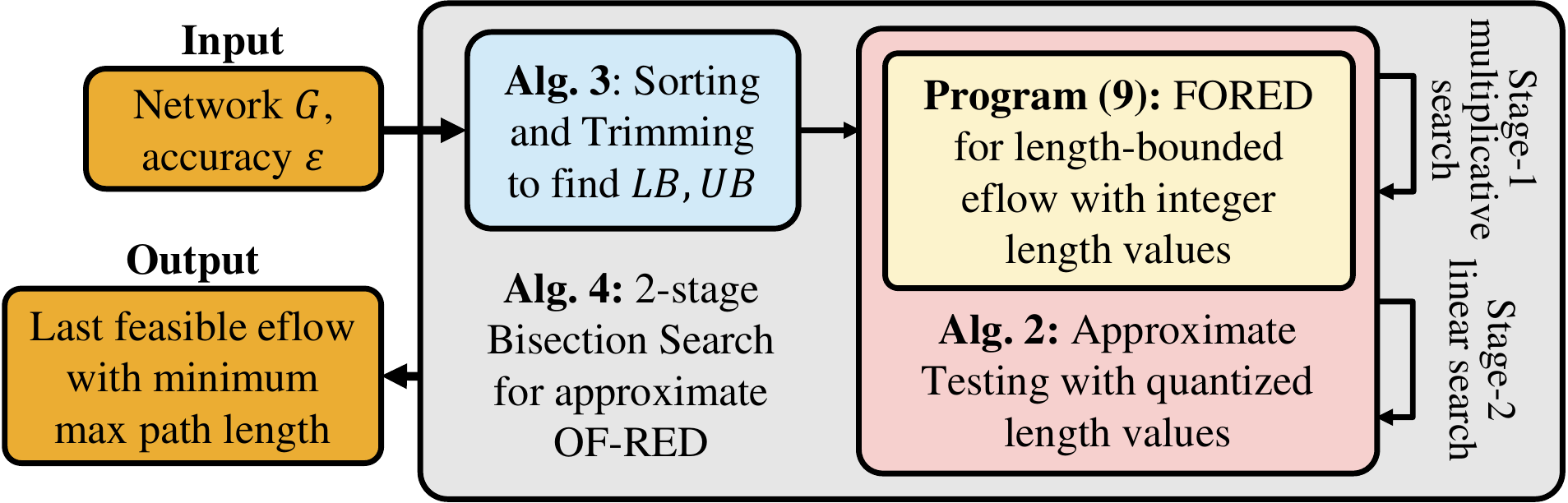}
\caption{\smallfont The overall algorithmic framework of FENDI.}
\label{fig:algorithmic_framework}
\end{figure}

The overall algorithmic framework, named \textbf{FENDI}, is shown in Fig.~\ref{fig:algorithmic_framework}.
Given an approximation parameter $\varepsilon > 0$, our FPTAS can obtain a $(1+\varepsilon)$-approximation to the optimal longest path length in time polynomial to the network graph size $|N|$ and $1/\varepsilon$.
Next, we design these building blocks one-by-one.

\subsection{Fidelity-aware Optimal Remote Entanglement Distribution}
\noindent
In the first building block, we first consider an opposite problem to OF-RED: maximizing expected EDR, subject to a fidelity bound that is equivalent to a path length bound $\mathbf{Z}$.
We address this problem in a very restrictive case: when all the length values $\zeta_l$ and $\zeta_n$ are \emph{positive integers}.
In this case, we can assume the path length bound is also a positive integer without loss of generality, which we instead denote as $Z$ to differentiate from a general, possibly non-integral path length $\mathbf{Z}$.

\textbf{Length-bounded eflow.}
The key to solving this ``integral'' problem optimally is to build the integer length values into the structure of the induced graph $\mathcal{G}$ of an eflow.
Let $[Z] = \{ 0, 1, 2, \dots, Z \}$.
Consider two enodes $mk$ and $kn$, whose ebits might be swapped to generate ebits between $mn$.
Depending on how the $mk$- and $kn$-ebits are generated, we can divide the two enodes each into $Z+1$ copies, which we denote as \emph{extended enodes} $mk/z$ and $kn/z$, for $z \in [Z]$.
Each enode $mk/z$ denotes $mk$-ebits that are generated along a path with path length of exactly $z$.
Because of the integer length bound $Z$, there are up to $Z+1$ different path length values (or equiv.\ $Z+1$ fidelity values) for ebits generated between each enode $mn$.
When two enodes $mk/z_1$ and $kn/z_2$ swap, if the resulting length $z_1 + z_2 + \zeta_k > Z$, the resulting ebits will not satisfy the length/fidelity bound, and hence should be discarded.
For elementary ebit generation, the initial enode is $mn/\zeta_{mn}$ if $\zeta_{mn} \le Z$, reflecting the initial fidelity of the elementary ebits on link $mn \in L$.

Fig.~\ref{fig:quantization} visualizes this transformation with a simple example.
Assume we have a three-node network shown in Fig.~\ref{fig:quantization}(a), and the goal is to establish $AC$-ebits either directly or with the help of repeater $B$.
Length values are marked beside nodes/links.
Given a length bound $Z = 6$, direct generation along link $AC$ would not be feasible with $\zeta_{AC} = 8$, and hence there is no extended enode $AC/8$ in Fig.~\ref{fig:quantization}(c).
Meanwhile, the feasible eflow of swapping $AB$ and $BC$ to generate $AC$ is visualized in Fig.~\ref{fig:quantization}(b)--(c), with the edges from extended enodes $AB/2$ and $BC/1$ to $AC/4$ given $\zeta_{AB} + \zeta_{BC} + \zeta_{B} = 2 + 1 + 1 = 4$.
\begin{figure}[t]
\centering
\includegraphics[width=0.485\textwidth]{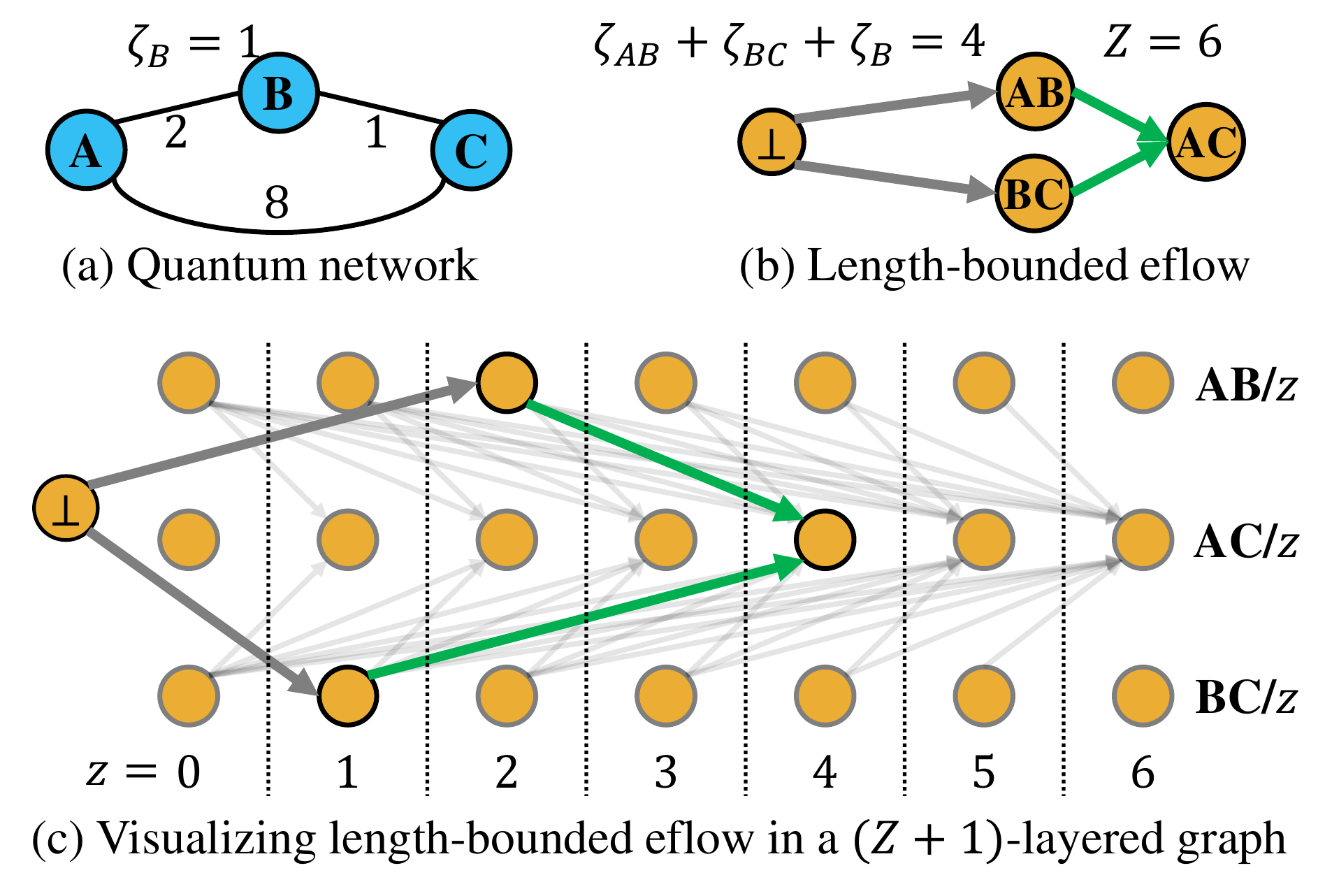}
\caption{Example of a length-bounded eflow in a 3-node network in (a). $AC$ is the SD pair. Length values $\zeta_l$ are marked beside links. Given a length bound $Z = 6$, there is one length-bounded eflow in (b). The eflow can be visualized in a $(Z+1)$-layered graph with all possible $g$ and $f$ variables as edges in (c).}
\label{fig:quantization}
\end{figure}

\textbf{FORED formulation.}
Building atop the above intuition, we extend ORED to FORED, whose solution (if feasible) is a length-bounded eflow achieving maximum expected EDR.
We keep the $g_{m{{{{{}}}}}n}$ variables unchanged for $mn \in L$. 
For each $f^{m{{{{{}}}}}k}_{m{{{{{}}}}}n}$ variable, we extend it to up to $O(Z^2)$ copies, denoted by $f^{m{{{{{}}}}}k/z'}_{m{{{{{}}}}}n/z}$, for $z' = [Z \!-\! \zeta_k]$, and $z = z' \!+\! \zeta_k, \dots, Z$.
In plain words, $f^{m{{{{{}}}}}k/z'}_{m{{{{{}}}}}n/z}$ denotes the \emph{number of $m{{{{{}}}}}k$-ebits, with a path segment length of $z'$, which contribute to swapping at node $k$ to generate $m{{{{{}}}}}n$-ebits with a path segment length of $z$}.
We then formulate FORED in Program~\eqref{fml:qofred}:
\begin{subequations}
\label{fml:qofred}
\begin{align}
    \max_{f, g} &\quad 
    \eta^Z_{st} \triangleq \sum\nolimits_{z=0}^{Z} I(st/z) 
    \tag{\ref{fml:qofred}} 
    \\   
    \text{s.t.}
    &\quad
    f^{m{{{{{}}}}}k/z_1}_{m{{{{{}}}}}n/z} = f^{k{{{{{}}}}}n/z_2}_{m{{{{{}}}}}n/z},
    \quad \forall m,n,k \in N,
    \nonumber\\
    &\quad\quad\quad
    \forall z_1, z_2 \!\in \! [Z\! -\! \zeta_k], 
    z \!=\! z_1 \!+\! z_2 \!+\! \zeta_k;
     \tag{\ref{fml:qofred}{a}} \label{fml:qofred:a}\\
    &\quad
    I(mn/z) = \Omega(mn/z), 
    \;\;\;\;
    \nonumber\\
    &\quad\quad\quad
    \forall z \in [Z], \forall m, n \!\in\! N, mn \!\ne\! st;
     \tag{\ref{fml:qofred}{b}}\label{fml:qofred:b}\\
    &\quad
    I(st/z) = 0, \quad\quad \forall z; \tag{\ref{fml:qofred}{c}}\label{fml:qofred:c}
\end{align}
where for $\forall m, n \in N$, $z \in [Z]$,
\label{eq:I2}
\begin{align}
    I(mn/z) &\triangleq q_{mn} c_{mn} g_{m{{{{{}}}}}n} \cdot \mathbf{1}_{mn \in L, \zeta_{mn} = z} 
    \notag \\
    &
    +\!\!\!\!\!\sum_{{k \in  N \atop \setminus \{m, n\}}} 
    \!\!\sum_{z' = 0}^{z-\zeta_k}
    \frac{q_k}{2} \left( f^{m{{{{{}}}}}k/z'}_{m{{{{{}}}}}n/z} + f^{k{{{{{}}}}}n/(z-z'-\zeta_k)}_{m{{{{{}}}}}n/z} \right)  \tag{\ref{eq:I2}{d}}\label{eq:I2:a},
    \\
    \Omega(mn/z) &\triangleq \!\!\!\! \sum_{k \in N \atop\setminus \{m, n\}} 
        \!\!\left( \;\;\; \sum_{\mathclap{z' = z+\zeta_n}}^{Z}
        f^{m{{{{{}}}}}n/z}_{m{{{{{}}}}}k/z'}
        \!\!+\!\!\sum_{\mathclap{z' = z+\zeta_m}}^{Z}f^{m{{{{{}}}}}n/z}_{k{{{{{}}}}}n/z'} \right) , \tag{\ref{eq:I2}{e}}\label{eq:I2:b}
\end{align}
\end{subequations}
and $\mathbf{1}_{mn \in L, \zeta_{mn} = z}$ denotes whether both $mn \in L$ and $\zeta_{mn} = z$.

\textbf{\emph{Explanation:}} 
Each constraint in Program~\eqref{fml:qofred} corresponds to one constraint in Program~\eqref{fml:ored}, applied to each extended enode.
Objective~\eqref{fml:qofred} is to maximize the sum of end-to-end ebits generated over all paths of lengths up to the bound $Z$, represented by enodes $st/z$ for $z \in [Z]$.
Constraint~\eqref{fml:qofred:a} considers the joint contribution to $m{{{{{}}}}}n$-ebits with a specific path length $z$, from a pair of $m{{{{{}}}}}k$- and $k{{{{{}}}}}n$-ebits with total path length $z - \zeta_k$.
This accounts for the fact that a concatenated $mn$-path has a total length of the $mk$-segment and the $kn$-segment, plus the length $\zeta_k$ of node $k$.
Constraint~\eqref{fml:qofred:b} specifies flow conservation at each intermediate pair of nodes $m{{{{{}}}}}n$ with each specific path length value $z$.
Constraint~\eqref{eq:I2:a} is the definition of $I(mn/z)$ that includes all generated ebits between $m$ and $n$ with a specific length $z$ from either elementary ebit generation or intermediate swapping, minus all ebits contributed to further swapping.
Constraint~\eqref{eq:I2:b} defines $\Omega(mn/z)$ that includes all the ebits between $m$ and $n$ with a specific length $z$ which will be swapped to build ebits between other node pairs.
\begin{theorem}
\label{th:qofred}
Given integer link/node lengths $\zeta_i > 0$ for $i \in N \cup L$, and an integer length bound $Z$, Program~\eqref{fml:qofred} computes the maximum expected EDR between $s$ and $t$, with all ebits generated along paths satisfying the length bound $Z$.
\myendbox
\end{theorem}
\begin{proof}
We call a $m{{{{{}}}}}n/z$ by \emph{enode $m{{{{{}}}}}n$ at level $z$}.
We first examine path length feasibility, \kIE, ebits generated between $m{{{{{}}}}}n$ at level $z$ has path length of exactly $z$.
For any physical link $mn \!\in\! L$, the first term in Eq.~\eqref{eq:I2:a} ensures that $g_{m{{{{{}}}}}n}$ only contributes to $I(mn/z)$ when $z = \zeta_{(m,n)}$, \kIE, elementary ebits along $mn$ are only counted at level $\zeta_{(m,n)}$.
Then, for any triple $m{{{{{}}}}}n/z$ where there exists $k$ and $z_1$ such that $f^{m{{{{{}}}}}k/z_1}_{m{{{{{}}}}}n/z} \!>\! 0$ and $f^{k{{{{{}}}}}n/z_2}_{m{{{{{}}}}}n/z} \!>\! 0$ (where $z = z_1+z_2+\zeta^{\theta}_k$), we can see that if ebits at $m{{{{{}}}}}k/z_1$ have path length of exactly $z_1$ and ebits at $k{{{{{}}}}}n/z_2$ have path length of exactly $z_2$, then ebits generated at $m{{{{{}}}}}n/z$ by swapping them at $k$ exactly have path length of $z = z_1 + z_2 + \zeta^{\theta}_k$.
By induction, any generated ebit at level $z$ has path length of exactly $z$.
Since there are at most $Z$ levels, all ebits generated between $s{{{{{}}}}}t$ have path lengths bounded by $Z$.

Next we prove optimality of Program~\eqref{fml:qofred}, by showing that every solution to Program~\eqref{fml:qofred} with objective value $\eta^{Z}_{st}$ is a solution to HF-RED with EDR bound $\Delta_{st} = \eta^{Z}_{st}$ and fidelity bound $\Upsilon_{st} = \frac{1}{4} \cdot \left( 1 + 3 e^{-Z} \right)$, and vice versa.
A feasible length-bounded eflow to Program~\eqref{fml:qofred} indicates a feasible eflow to Program~\eqref{fml:ored}, by summing up $f$ variables and $I(\cdot)$ function values over all possible $z$.
Combined with path length feasibility, the length-bounded eflow maintains worst-case fidelity above the fidelity threshold $\Upsilon_{st}$ and EDR bound $\Delta_{st}$ in the HF-RED problem.
Now, for a feasible length-bounded eflow, let us represent it by a set of pflows with induced graphs $\{ \mathcal{G} \}$ and assigned values $\{ \eta_{\mathcal{G}} \}$.
Each $\mathcal{G}\! =\! (\mathcal{V}, \mathcal{E})$ would represent a path $p_\mathcal{G} \in G$ with path length bounded by $Z$.
We can construct a feasible solution to Program~\eqref{fml:qofred} given each $\mathcal{G}$.
For each enode $m{{{{{}}}}}n \in \mathcal{G}$, let $\zeta_{m{{{{{}}}}}n}$ be the length of the path segment in $G$ between $m$ and $n$ that is represented by $\mathcal{G}$ (which can be computed for each $m{{{{{}}}}}n$ in linear time).
For each enode $m{{{{{}}}}}n \in \mathcal{V}$ that has no in-coming link, we set $g_{m{{{{{}}}}}n} \!\!= \!\bar g_{m{{{{{}}}}}n}\! \cdot \eta_{\mathcal{G}}$.
Then, for each $(m{{{{{}}}}}k, m{{{{{}}}}}n) \in \mathcal{E}$, we set $f^{m{{{{{}}}}}k/\zeta_{m{{{{{}}}}}k}}_{m{{{{{}}}}}n/\zeta_{m{{{{{}}}}}n}} = f^{k{{{{{}}}}}n/\zeta_{k{{{{{}}}}}n}}_{m{{{{{}}}}}n/\zeta_{m{{{{{}}}}}n}} = \bar f^{m{{{{{}}}}}k}_{m{{{{{}}}}}n} \cdot \eta_{\mathcal{G}}$.
It can be checked that the constructed solution is feasible to Program~\eqref{fml:qofred} based on how $\{ \bar g_{m{{{{{}}}}}n}, \bar f^{m{{{{{}}}}}k}_{m{{{{{}}}}}n}, \bar f^{k{{{{{}}}}}n}_{m{{{{{}}}}}n} \}$ are computed, how $\mathcal{G}$ is defined, and that each $\mathcal{G}$ represents a path with length bounded by $Z$.
Summing up so-constructed solutions for all of $\{ \mathcal{G} \}$ and $\{ \eta_{\mathcal{G}} \}$, we get a feasible solution to Program~\eqref{fml:qofred}, with the same objective value $\eta^Z_{st} \!= \!\!\sum_{\mathcal{G}} \eta_{\mathcal{G}}$.
It follows that Program~\eqref{fml:qofred} outputs the maximum expected EDR among all feasible eflows satisfying the path length bound of $Z$.
\end{proof}

\begin{proposition}
\label{prop:qofred-1}
Program~\eqref{fml:qofred} can be solved optimally, in time polynomial to the input size and $Z$.
\myendbox
\end{proposition}
\begin{proof}
Program~\eqref{fml:qofred} is an LP with $O(|N|^3Z^2)$ variables, and can be solved in time polynomial to $|N|$ and $Z$~\cite{Ye1991}.
\end{proof}

\begin{algorithm}[t]
\smallfont
\caption{\mbox{Approximate testing procedure $\text{TEST}(\mathbf{Z}, \varepsilon)$}}
\label{a:test}
\KwIn{Network $G$, accuracy $\varepsilon$, non-quantized length bound $\mathbf{Z}$}
\KwOut{Test result $\varsigma \in \{ true, false \}$}
$\theta \leftarrow (2|N|-3)/ (\varepsilon \mathbf{Z})$, and $Z \leftarrow \lfloor \theta \mathbf{Z} \rfloor + (2|N|-3)$\;\label{a:test:ln:val}
Solve Program~\eqref{fml:qofred} with $\{ \zeta_i^\theta \}$ and $Z$, and get $\eta^{Z}_{st}$\;
\Return{((Program~\eqref{fml:qofred} is feasible) AND ($\eta^{Z}_{st} \ge \Delta_{st}$)).}
\end{algorithm}

\subsection{Approximate Testing Procedure}
\noindent
Program~\eqref{fml:qofred} runs in pseudo-polynomial time and can be used to check, given any length bound $Z$, if there is a feasible length-bounded eflow whose expected EDR can satisfy an EDR bound $\Delta_{st}$.
This testing is however limited by 1) the requirement in Program~\eqref{fml:qofred} that all length values must be positive integers, and 2) the pseudo-polynomial running time.
In this subsection, we design an \emph{approximate testing} procedure which simultaneously addresses these two issues.
Specifically, by designing a proper quantization scheme to transform any real length value into a positive integer within a polynomial scale, we can both limit the size of the resulting LP in Program~\eqref{fml:qofred}, and bound the quantization error introduced by the transformation.

To start, we define a quantization of the length values $\mathcal{Z} \triangleq \{ \zeta_i \,|\, i \in L \cup N \}$ with a factor $\theta > 0$, where the \textbf{quantized length} is denoted by $\zeta^{\theta}_i \!=\! \lfloor \theta \cdot \zeta_i \rfloor \!+\! 1$, for $i \in N \cup L$.
This transformation ensures that the resulting value $\zeta_i$ is always a positive integer, which satisfies the requirement of Program~\eqref{fml:qofred}.

Let $\zeta^{\theta}(p)$ be the length of an arbitrary path $p$ in $G$ after quantization with factor $\theta$, and recall that $\zeta(p)$ is the original path length.
We have the following lemma:
\begin{lemma}
\label{l:quant}
    $\theta \cdot \zeta(p) \le \zeta^{\theta}(p) \le \lfloor \theta \cdot \zeta(p) \rfloor + (2|N| - 3)$.
\myendbox
\end{lemma}
\begin{proof}
The left side is trivial due to how lengths are quantized. The right side is because 1) each entanglement path in $G$ has at most $|N|\!-\!1$ links and $|N|\!-\!2$ intermediate nodes whose lengths are counted (excluding source and destination), and 2) $\zeta^{\theta}(p)$ is an integer value due to quantization (and hence the floor over $\theta \cdot \zeta(p)$ on the right side).
\end{proof}
Based on Lemma~\ref{l:quant}, we design the approximate testing procedure in Algorithm~\ref{a:test}.
Suppose an accuracy parameter $\varepsilon > 0$ and a non-quantized length bound $\textbf{Z}$ are given, and define quantization factor $\theta$ and corresponding quantized length bound $Z$ in Line~\ref{a:test:ln:val}.
The algorithm returns a test result $\varsigma \in \{ true, false \}$, which indicates whether the network admits a feasible \emph{length-bounded eflow} with expected EDR no lower than the EDR bound $\Delta_{st}$.
Let $\textbf{Z}^*$ be the non-quantized length of the optimal solution of the original OF-RED problem.
Lemma~\ref{l:test} shows a numerical relationship between the input length bound $\mathbf{Z}$ and the optimal $\textbf{Z}^*$ given the testing outcome:
\begin{lemma}
\label{l:test}
Given any $\varepsilon > 0$ and $\mathbf{Z} > 0$, we have
\begin{align*}
\pushQED{\qed} 
&\text{\normalfont TEST}(\mathbf{Z}, \varepsilon) = true &&\Rightarrow&& \mathbf{Z}^* \le (1 + \varepsilon) \cdot \mathbf{Z};	\\
&\text{\normalfont TEST}(\mathbf{Z}, \varepsilon) = false &&\Rightarrow&& \mathbf{Z}^* > \mathbf{Z}. \qedhere
\popQED
\end{align*}
\end{lemma}
\begin{proof}
If $\text{\normalfont TEST}(\mathbf{Z}, \varepsilon) = true$, we have a feasible length-bounded eflow with maximum EDR $\eta^{Z}_{st} \ge \Delta_{st}$ and all paths satisfying bound $Z$.
This translates to a feasible solution to OF-RED with EDR bound $\Delta_{st}$.
Let $p$ be the maximum-length path in the solution  w.r.t. the original lengths $\mathcal{Z}$.
Following Lemma~\ref{l:quant}, we have:
\begin{align*}
    \zeta(p)    &\le \zeta^{\theta}(p) / \theta \le Z / \theta 
                \le (1+\varepsilon)\mathbf{Z}.
\end{align*}
Since the solution is feasible to OF-RED, its maximum (non-quantized) path length is an upper bound on $\mathbf{Z}^*$, and hence we have $\mathbf{Z}^* \le (1 + \varepsilon) \mathbf{Z}$.
This proves the first statement.

To prove the second statement, we show that as long as there is a feasible OF-RED solution which has maximum path length bounded by $\mathbf{Z}$, then $\text{TEST}(\mathbf{Z}, \varepsilon)$ must return $true$.
Consider such a solution for which every path $p$ satisfies that $\zeta(p) \le \mathbf{Z}$.
By Lemma~\ref{l:quant}, we have:
\begin{align*}
    \zeta^{\theta}(p)   \le \theta \cdot \zeta(p) + (2|N|-3) 
                        \le {(2|N|-3)}/\varepsilon + (2|N|-3).
\end{align*}

Since $\zeta^{\theta}(p)$ must be an integer, this implies $\zeta^{\theta}(p) \le \lfloor {(2|N|-3)}/\varepsilon \rfloor + (2|N|-3) = \lfloor \theta \textbf{Z} \rfloor + (2|N|-3) = Z$.
By Theorem~\ref{th:qofred}, this solution can be decomposed into a set of pflows, whose maximum quantized path length is $\zeta^{\theta}(p)$, and whose sum of objective values equals $\eta^Z_{st} \ge \Delta_{st}$.
In this case, $\text{TEST}(\mathbf{Z}, \varepsilon)$ must return $true$.
Hence if $\text{TEST}(\mathbf{Z}, \varepsilon)$ returns $false$, it indicates there is no such feasible solution.
\end{proof}

\textbf{\emph{Remark:}}
The choice of the quantization factor $\theta$ in Line~\ref{a:test:ln:val} is key to ensuring both a polynomial size and bounded quantization error.
On one hand, it ensures that the quantized length bound $Z$ is polynomial to $|N|/\varepsilon$ regardless of the value of the original length bound $\mathbf{Z}$.
On the other hand, utilizing the maximum path length in the network, it ensures that the testing result has an error of at most $(1+\varepsilon)$.

The testing procedure is designed to enable a bisection search for the minimum longest path length $\mathbf{Z}^*$, if a reasonable initial range $[\text{LB}, \text{UB}]$ of $\mathbf{Z}^*$ is given.
By repeatedly testing if a length bound $\mathbf{Z} \in [\text{LB}, \text{UB}]$ is feasible or not, the search can multiplicatively reduce the search space, and return a close-to-optimal feasible length bound $\mathbf{Z}$ within time logarithmic to the size of the initial search space.
Since the time complexity of the search depends on the size of the search space, we next seek to find a pair of lower bound $\text{LB}$ and upper bound $\text{UB}$ on the optimal $\mathbf{Z}^*$ that are reasonably close to each other.

\subsection{Sorting and Trimming Algorithm}
\noindent
We design a sorting and trimming algorithm in Algorithm~\ref{a:lbub} to find an initial pair of bounds $\text{LB}, \text{UB}$ on $\mathbf{Z}^*$, such that $\text{LB} \le \mathbf{Z}^* \le \text{UB}$.
Algorithm~\ref{a:lbub} sorts all node/link lengths in descending order, and then tries to find a \emph{critical length} $\zeta_{[i-1]}$ such that $G_{[i-1]}$ still admits a feasible solution to Program~\eqref{fml:ored} with $\eta^{[i-1]}_{st} \ge \Delta_{st}$, but $G_{[i]}$ does not.
This means at least one node/link with length no less than $\zeta_{[i-1]}$ is needed to satisfy the EDR bound of $\Delta_{st}$.
Consequently, the optimal $\mathbf{Z}^*$ must be at least $\zeta_{[i-1]}$ as a lower bound.
Besides, since there is a feasible solution in $G_{[i-1]}$, and each path can have at most $|N|-1$ links and $|N|-2$ intermediate nodes, the feasible solution has a maximum path length of $(2|N|-3) \cdot \zeta_{[i-1]}$ as all nodes and links in $G_{[i-1]}$ have lengths at most $\zeta_{[i-1]}$.
This shows that $(2|N|-3) \cdot \zeta_{[i-1]}$ is an upper bound on $\mathbf{Z}^*$.
The gap between the above pair of bounds is a multiplicative factor of $\text{UB} / \text{LB} = 2|N|-3 \in O(|N|)$.

\begin{algorithm}[t]
\smallfont
\caption{\mbox{Finding lower and upper bounds on $\mathbf{Z}^*$}}
\label{a:lbub}
\KwIn{Network $G$}
\KwOut{Lower and upper bounds $(\text{LB}, \text{UB})$ on $\mathbf{Z}^*$}
Sort node/link lengths in $\{ \zeta_l \,|\, l \in L \} \cup \{ \zeta_n \,|\, n \in N \}$ in descending order as $\mathcal{Z} = (\zeta_{[1]}, \zeta_{[2]}, \dots)$\;
\For{$\zeta_{[i]} \in \mathcal{Z}$ in sorted order}{
    Construct graph $G_{[i]}$ by pruning all nodes and links with lengths greater than $\zeta_{[i]}$ in $G$\;
    Solve Program~\eqref{fml:ored} on $G_{[i]}$ for $\eta^{[i]}_{st}$\;
    \lIf{Infeasible or $\eta^{[i]}_{st} < \Delta_{st}$}{\textbf{break}}
}
\Return{$(\text{\normalfont LB} = \zeta_{[i-1]}, \text{\normalfont UB} = (2|N|-3)\zeta_{[i-1]})$.}
\end{algorithm}

\subsection{Two-stage bisection search Algorithm}
\noindent
\noindent
After finding $\text{LB}$ and $\text{UB}$ with Algorithm~\ref{a:lbub}, we can apply a bisection search on the range $[\text{LB}, \text{UB}]$ to find an approximator of $\mathbf{Z}^*$.
Each time we define a bound $\mathbf{Z}\! =\! (\text{LB} \! + \! \text{UB})/2$, and call $\text{TEST}(\mathbf{Z}, \varepsilon)$.
If $\text{TEST}(\mathbf{Z}, \varepsilon)$ outputs $true$, we narrow the gap by setting $\text{UB} \!\leftarrow\! (1 + \varepsilon) \mathbf{Z}$; otherwise, we set $\text{LB}\! \leftarrow \!\mathbf{Z}$.
To achieve the desired accuracy, it takes at least $O(\log (\text{UB}\! -\! \text{LB}))\! =\! O(\log (|N|\zeta_{[i-1]}))$ search iterations (where $\zeta_{[i-1]}$ is the critical length in Algorithm~\ref{a:lbub}), each making a call to $\text{TEST}(\mathbf{Z}, \varepsilon)$ which solves an LP of size $O ( |N|^3 ( {|N|}/{\varepsilon} )^2 )$.

In Algorithm~\ref{a:fptas}, we propose an improved \emph{2-stage search algorithm}, which reduces the asymptotic search complexity and sizes of the LPs solved in most search iterations.
In Stage-1 (Lines~\ref{a:fptas:ln:s1s}--\ref{a:fptas:ln:s1e}), a \emph{multiplicative bisection} (bisection in the logarithmic scale) is done on $[\text{LB}, \text{UB}]$, where each time an $\varepsilon = 1$ is used in approximate testing.
By Lemma~\ref{l:test}, $\text{\normalfont TEST}(\mathbf{Z}, 1)$ returning $false$ means $\mathbf{Z}^* > \mathbf{Z}$ and hence LB is increased to $\mathbf{Z}$; $\text{\normalfont TEST}(\mathbf{Z}, 1)$ returning $true$ means $\mathbf{Z}^* \le (1 + \varepsilon) \cdot \mathbf{Z} = 2\mathbf{Z}$ and hence UB is decreased to $2\mathbf{Z}$.
Stage-1 ends when LB and UB are within a constant factor of each other, such as $\text{UB}/\text{LB} \le 4$.

In Stage-2, instead of doing bisection directly on $[\text{LB}, \text{UB}]$, we do a bisection on the quantized bounds $[Z_{\text{LB}}, Z_{\text{UB}}]$.
We fix the quantization factor $\theta = (2|N|-3)/(\varepsilon \text{LB})$, and only vary the quantized path length bound $Z$.
The main purpose of this construction is to utilize quantization to naturally reduce the number of search iterations to achieve the desired accuracy defined by $\varepsilon$.
Since LB and UB are within a constant ratio of each other, the quantized length bound $Z_{\text{UB}} \in O(|N|/\varepsilon)$, and hence $O(\log (|N|/\varepsilon))$ search iterations are needed to search all integers between $Z_{\text{LB}}$ and $Z_{\text{UB}}$.
This makes the search complexity no longer related to the critical length $\zeta_{[i-1]}$ as in the naive bisection search.
Let $\mathbf{Z}^{\theta}$ be the minimum longest path length for \emph{quantized OF-RED (QOF-RED)} with $\theta$.
The lemmas below show this quantized bisection is as effective as the bisection search on the original bounds $[\text{LB}, \text{UB}]$. 

\begin{algorithm}[t]
\smallfont
\caption{\mbox{$2$-stage Bisection for Approximate OF-RED}}
\label{a:fptas}
\KwIn{Network $G$, search accuracy parameter $\varepsilon$}
\KwOut{Eflow with maximum path length $\mathbf{Z}^{+}$}
Call Algorithm~\ref{a:lbub} to find LB and UB on $\mathbf{Z}^*$\;\label{a:fptas:ln:1}
\While(\tcp*[f]{Stage-1}){$\text{\normalfont UB} > 4 \cdot \text{\normalfont LB}$}{\label{a:fptas:ln:s1s}
    $\mathbf{Z} = \sqrt{ (\text{\normalfont UB} \cdot \text{\normalfont LB} ) / {2}}$\;
    \lIf{$\text{\normalfont TEST}(\mathbf{Z}, 1) = false$}{$\text{\normalfont LB} \leftarrow \mathbf{Z}$}
    \lElse{$\text{\normalfont UB} \leftarrow 2 \cdot\mathbf{Z}$}
}\label{a:fptas:ln:s1e}
\mbox{$\theta \!\leftarrow\! \frac{2|N|-3}{\varepsilon \text{LB}}$, $Z_{\text{LB}} \!\leftarrow\! \lfloor \theta \text{LB} \rfloor$, $Z_{\text{UB}} \!\leftarrow\! \lfloor \theta \text{UB} \rfloor \!+\! (2|N| \!-\! 3)$}\;
\While(\tcp*[f]{Stage-2}){$Z_{\text{\normalfont UB}} > Z_{\text{\normalfont LB}} + 1$}{\label{a:fptas:ln:s2s}
    $Z \leftarrow \lfloor ( Z_{\text{LB}} + Z_{\text{UB}} ) / 2 \rfloor$\;
    Solve Program~\eqref{fml:qofred} with $\theta$ and $Z$, and get $\eta^{Z}_{st}$\;
    \lIf{Program~\eqref{fml:qofred} is feasible AND $\eta^{Z}_{st} \ge \Delta_{st}$}{$Z_{\text{UB}}\leftarrow Z$}
    \lElse{$Z_{\text{LB}}\leftarrow Z$}
}\label{a:fptas:ln:s2e}
\Return{last feasible solution with max path length $\mathbf{Z}^+$}
\end{algorithm}

\begin{lemma}
\label{l:quant1}
$\lfloor \theta \text{\normalfont LB} \rfloor \le \mathbf{Z}^{\theta} \le \lfloor \theta \text{\normalfont UB} \rfloor + (2|N|-3)$.
\myendbox
\end{lemma}
\begin{lemma}
\label{l:quant2}
$\mathbf{Z}^{\theta} \le \theta \cdot (1 + \varepsilon) \cdot \mathbf{Z}^*$.
\myendbox
\end{lemma}
\begin{proof}
Note that a feasible solution to OF-RED indicates a feasible solution to QOF-RED, and vice versa.
Given the optimal solution to original OF-RED with objective $\mathbf{Z}^*$, let $p$ be its longest entanglement path such that $\zeta(p) = \mathbf{Z}^*$, and let $p_{\theta}$ be its longest entanglement path with quantization.
By Lemma~\ref{l:quant}, $\zeta^{\theta}(p_{\theta}) \le \lfloor \theta \zeta(p_{\theta}) \rfloor + (2|N|-3) \le \lfloor \theta \zeta(p) \rfloor + (2|N|-3) \le \lfloor \theta \text{\normalfont UB} \rfloor + (2|N|-3)$.
This proves the right-hand side of Lemma~\ref{l:quant1}, as $\mathbf{Z}^{\theta}$ is optimal and hence $\mathbf{Z}^{\theta} \le \zeta^{\theta}(p_{\theta})$.
Further, since $\zeta(p) = \mathbf{Z}^*$, we have $\zeta^{\theta}(p_{\theta}) \le \theta \zeta(p) + (2|N|-3) = \theta (\mathbf{Z}^* + (2|N|-3)/\theta) = \theta (\mathbf{Z}^* + \varepsilon \text{\normalfont LB}) \le \theta \cdot (1 + \varepsilon) \cdot \mathbf{Z}^*$, and hence $\mathbf{Z}^{\theta} \le \zeta^{\theta}(p_{\theta}) \le \theta \cdot (1 + \varepsilon) \cdot \mathbf{Z}^*$.

Now consider the optimal solution to QOF-RED, and let $p'_{\theta}$ be its quantized longest entanglement path, where $\zeta^{\theta}(p'_{\theta}) = \mathbf{Z}^{\theta}$.
Let $p'$ be its longest entanglement path without quantization.
Since this solution is also feasible to OF-RED, we have $\theta \text{\normalfont LB} \le \theta \mathbf{Z}^* \le \theta \zeta(p')$.
By Lemma~\ref{l:quant}, we then have $\theta \zeta(p') \le \zeta^{\theta}(p') \le \zeta^{\theta}(p'_{\theta}) = \mathbf{Z}^{\theta}$.
Hence $\mathbf{Z}^{\theta} \ge \lfloor \theta \text{\normalfont LB} \rfloor$.
\end{proof}

Below, Theorem 4 states our main result. 

\begin{theorem}
\label{th:fptas}
Given accuracy parameter $\varepsilon$, Algorithm~\ref{a:fptas} finds a $(1+\varepsilon)$-approximation of the optimal OF-RED path length value $\mathbf{Z}^*$, within time polynomial to $|N|$ and $1/\varepsilon$.
\myendbox
\end{theorem}

\begin{proof}
The approximation ratio directly comes from Lemma~\ref{l:quant2}.
Let $T(x)$ be the time for solving an LP with $x$ variables.
First, Algorithm~\ref{a:lbub} finds $[\text{LB}, \text{UB}]$ on $\mathbf{Z}^*$ in up to $|\mathcal{Z} \!| =\! |N| + |L|$ iterations, each solving Program~\eqref{fml:ored} with $O(|N|^3)$ variables in $O(T(|N|^3))$ time.
For Stage-1 bisection of Algorithm~\ref{a:fptas}, let $\pi_{[j]}$ be the ratio $\text{UB} / \text{LB}$ after the $j$-th iteration.
Initially $\pi_{[0]}\! =\! 2|N|\!-3$ due to $[\text{LB}, \text{UB}]$ bound by Algorithm~\ref{a:lbub}.
After each iteration $j$, $\pi_{[j]}\! =\! \sqrt{2\pi_{[j-1]}}$ based on how $\mathbf{Z}$ is computed.
Let $J$ be index of the last iteration, and apply the above recursively, then we have $\pi_{[J]}\!\! =\! 2^{1/2+1/4+\cdots+1/2^J} \!\cdot\! \pi_{[0]}^{1/2^J} \!\!\le 2\! \cdot \pi_{[0]}^{1/2^J}\! =\! 2\! \cdot (2|N|-3)^{1/2^J}$.
As $\pi_{[J]} \le 4$ when Stage-1 ends, the total number of iterations is $O(\log \log |N|)$.
Each iteration solves Program~\eqref{fml:qofred} with $\varepsilon = 1$, and hence $Z \in O(|N|)$, resulting in $O(|N|^3Z^2) = O(|N|^5)$ variables.
Thus each iteration takes $O(T(|N|^5))$ time.
For Stage-2, the bisection is done on up to $Z_{\text{\normalfont UB}} \in O(\frac{|N|}{\varepsilon})$ integers, with up to $O(\log \frac{|N|}{\varepsilon})$ search iterations.
Each iteration solves Program~\eqref{fml:qofred} with $O(|N|^3Z_{\text{\normalfont UB}}^2) = O(\frac{|N|^5}{\varepsilon^2})$ variables, and thus takes $O(T(\frac{|N|^5}{\varepsilon^2}))$ time.
Summing up the above, the overall time complexity is $O( T(|N|^3) \cdot (|N| + |L|) + T(|N|^5) \cdot \log\log |N| + T(|N|^5 / \varepsilon^2) \cdot \log \frac{|N|}{\varepsilon} )$.
Since an LP can be solved in polynomial time~\cite{Ye1991}, the above time is polynomial to $|N|$ and $1/\varepsilon$.
\end{proof}

\subsection{Discussions}

\noindent
\textbf{\emph{Reducing running time:}}
Despite being polynomial-time, Algorithm~\ref{a:fptas} still has high complexity due to solving the large-size LPs.
There are several methods to reduce running time: 1) setting a loose $\varepsilon$; 
2) applying heuristic quantization that works empirically; 3) developing heuristic algorithms to solve the quantized LP.
We will examine effect of the first method in our evaluation.
Considering that a quantum network is designed for long-term operations, the overhead of offline optimization can often be negligible. 
For instance, by spending minutes or hours to compute a high-EDR and high-fidelity entanglement distribution plan for a quantum key distribution (QKD) application~\cite{peev2009secoqc}, the plan could be executed and deliver largely improved performance over a period of weeks or months before offline maintenance/re-optimization is needed.
We will explore efficient real-time protocol design in future research.

\textbf{\emph{Entanglement distribution protocol:}}
While the goal of our algorithm is mainly to 1) compute theoretical upper bounds on the achievable EDR and worst-case fidelity and 2) characterize the EDR-fidelity trade-off, we note that the computed eflow can actually be implemented by a data plane protocol as shown in the Appendix.
To achieve the theoretical EDR and fidelity, quantum memories are required for performing post-selection and storage before further swapping.
In evaluation, we will use this protocol to characterize the EDR-fidelity trade-off in a simulated quantum network, and evaluate the performance of several state-of-the-art protocols with respect to the characterized trade-off.

\textbf{\emph{Entanglement purification and error correction:}}
This paper does not consider quantum operations that may improve fidelity during entanglement distribution, such as purification or quantum error correction (QEC).
Both purification and QEC require consuming multiple/many additional ebits or qubits in order to get one high-quality ebit.
This may significantly reduce the achievable EDR.
Both operations also require idealized quantum memories not only for storage but also for local quantum computation, which are far more complicated to design and implement.
With the abstractions developed in this paper, we wish to explore incorporating purification and QEC into end-to-end modeling in our future work.
%
%





\section{Performance Evaluation}
\label{pe}

\subsection{Evaluation Methodology}
\label{sec:em}
\noindent
To evaluate the performance of our proposed algorithm, we developed a discrete-time quantum network simulator and carried out simulations on different randomly generated topologies.
We used random Waxman graphs~\cite{waxman} with parameters $\alpha=\beta=0.8$.
Each node or link had a success probability of $0.5$ and $0.9$, respectively, and fidelity uniformly sampled from $[0.7, 0.95]$.
Each link had a capacity uniformly sampled from $[26, 35]$.
Parameters were selected as the same values as in existing work~\cite{zhao2022e2e}, except for the swapping success probability, which should not exceed $0.5$ due to the limitation of current BSM scheme with linear optics~\cite{bayerbach2023bell}.
In each setting, we generated $5$ graphs each with $15$ nodes and $3$ random SD pairs, except in Fig.~\ref{fig:trade_off} where we characterized the entire trade-off curve for one SD pair in a single graph.
Results were averaged over all runs in the same setting to average-out random noise.

Our simulator was based on a time-slotted model to be compatible to existing algorithms, though our data plane protocol (see Appendix) does not require network-wide synchronization.
Linear programs were solved by Gurobi~\cite{gurobi}.
Simulations were ran on a Linux desktop with a $12$-core $4$GHz CPU and $256$GB memory.
In each simulation, we first ran our proposed FPTAS algorithm or a comparison algorithm for the SD pair.
Based on the solution, we then simulated entanglement generation, swapping and/or queuing for {$1000$} time slots.
The following entanglement routing/distribution algorithms were compared:
\begin{itemize}
    \item \textbf{FENDI:}
    Our proposed FPTAS, with the solution executed using the post-selection-and-storage protocol in Appendix.
    \item \textbf{ORED:}
    The fidelity-agnostic ORED algorithm, with a similar post-selection-and-storage protocol in~\cite{Dai2020a}.
    \item \textbf{E2E-F:}
    End-to-end fidelity-aware entanglement routing in~\cite{zhao2022e2e}, \emph{without purification} for fair comparison.
    \item \textbf{QPASS:}
    Fidelity-agnostic entanglement routing in~\cite{shi2020concurrent}.
\end{itemize}
For our algorithm, we set $\varepsilon = 0.5$ by default.
For QPASS and E2E-F, we set the number of paths $K = 30$.
Since E2E-F and QPASS are \emph{entanglement routing} algorithms for a bufferless quantum network, we adapted our simulator to discard all saved ebits after one time slot when simulating them.

The following metrics were used for evaluation.
The \textbf{\emph{minimum fidelity}} and \textbf{\emph{average fidelity}} measure the lowest and average fidelity values of all end-to-end entanglements.
The \textbf{\emph{EDR satisfaction ratio}} measures the fraction of simulation runs where the EDR bound is met.
The \textbf{\emph{running time}} measures the average time spent on running each \emph{control plane} algorithm.

\begin{figure}[t]
\centering
\includegraphics[width=0.3\textwidth]{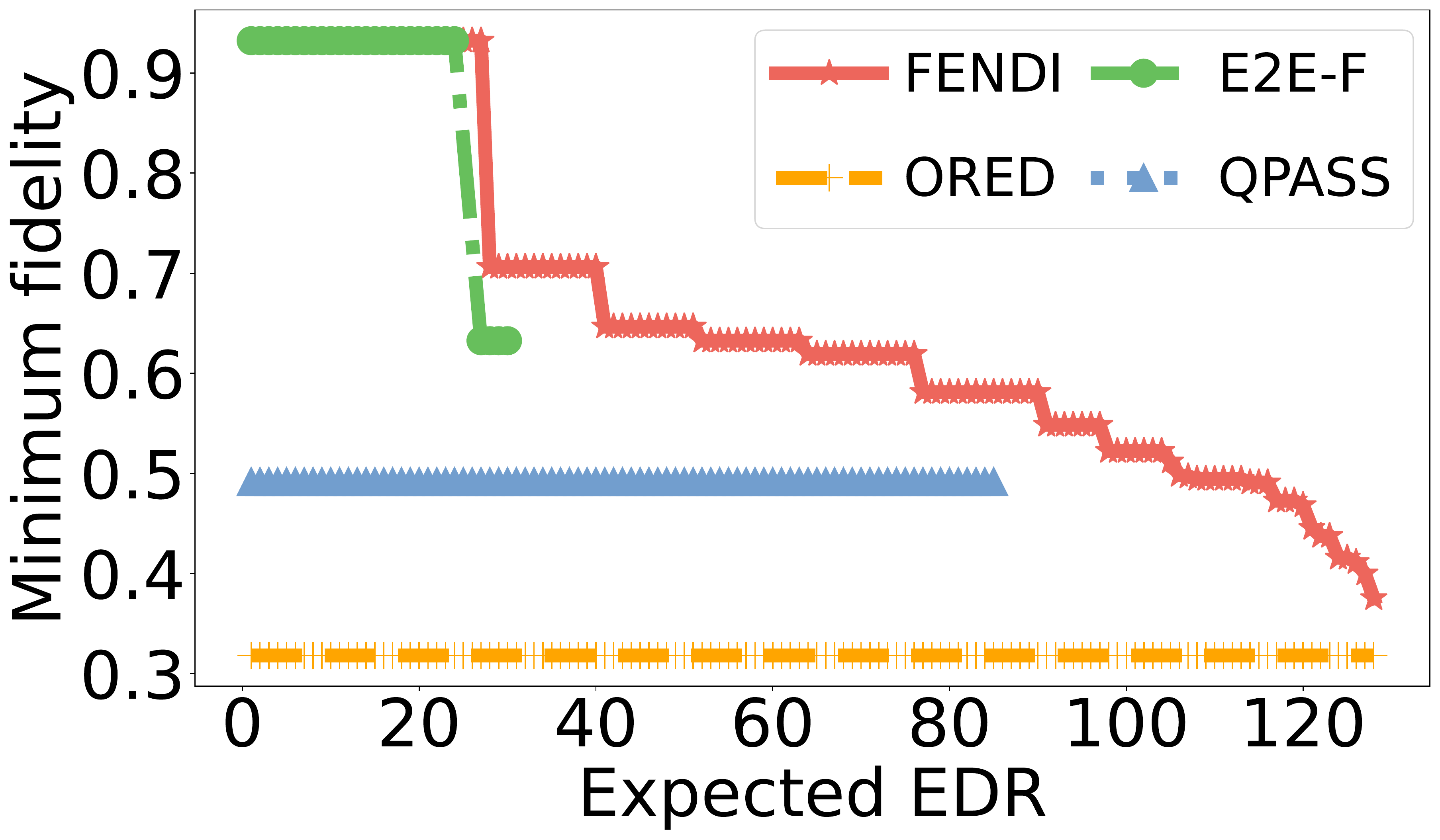}
\caption{The trade-off between worst-case fidelity and expected EDR for compared algorithms.}
\label{fig:trade_off}
\end{figure}

\begin{figure*}[t]
\vspace{-1em}
\centering
\subfloat[Minimum end-to-end fidelity]{\includegraphics[width=0.3\textwidth]{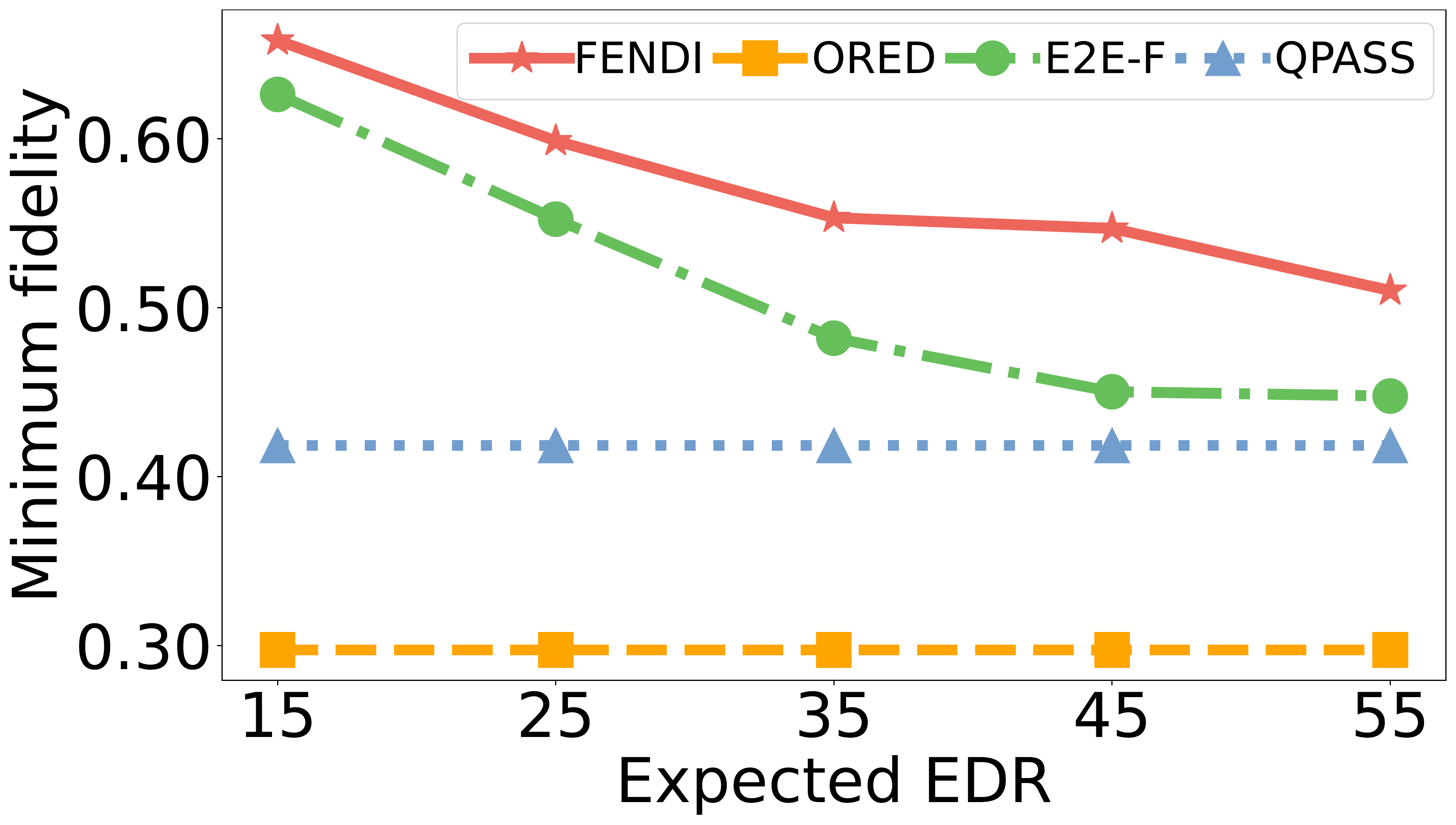}
\label{fig:algs_comp:algs_worst_fidelity}}
\hfil
\subfloat[Average end-to-end fidelity]{\includegraphics[width=0.3\textwidth]{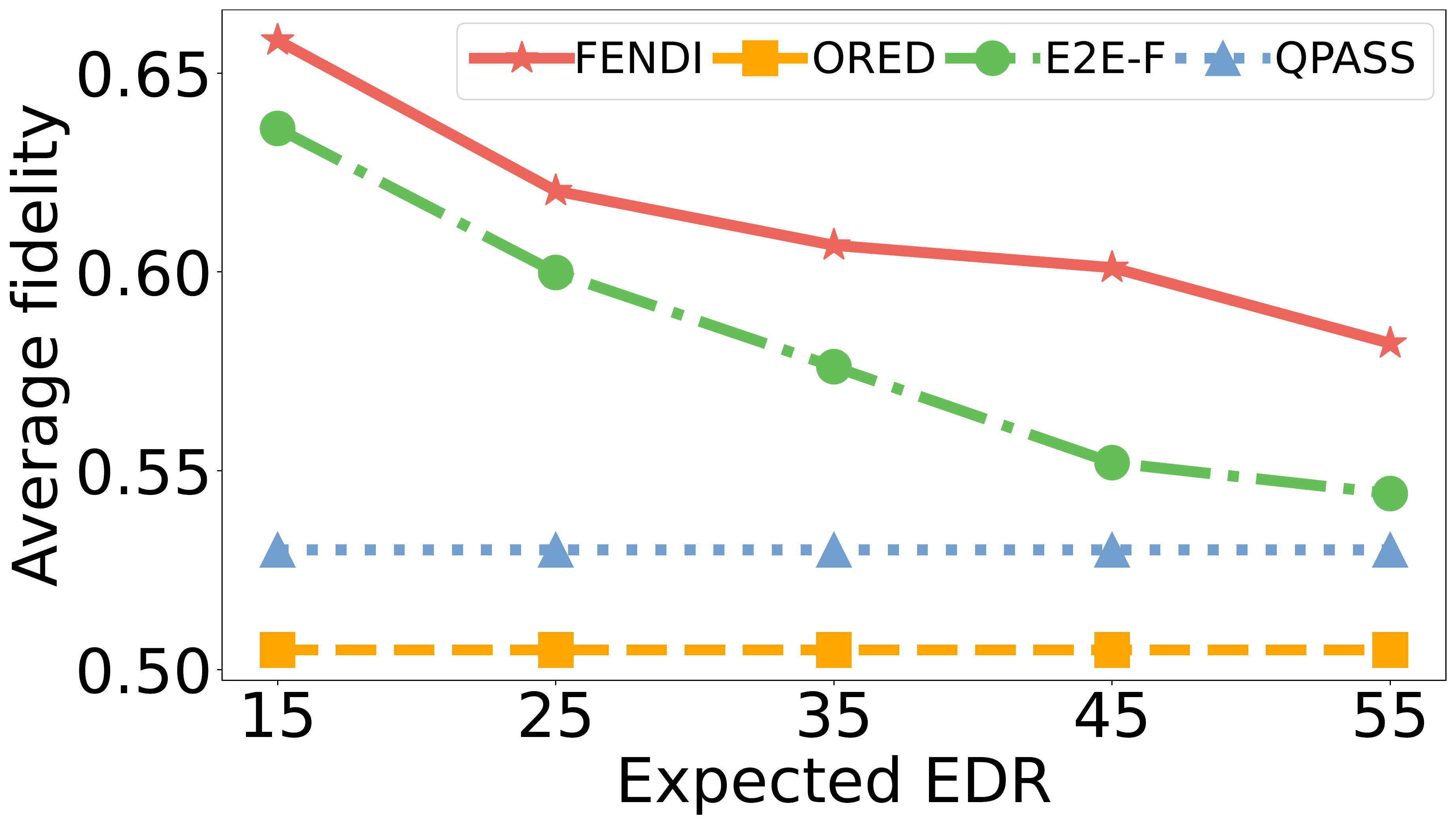}
\label{fig:algs_comp:algs_average_fidelity}}
\hfil
\subfloat[EDR satisfaction ratio]{\includegraphics[width=0.3\textwidth]{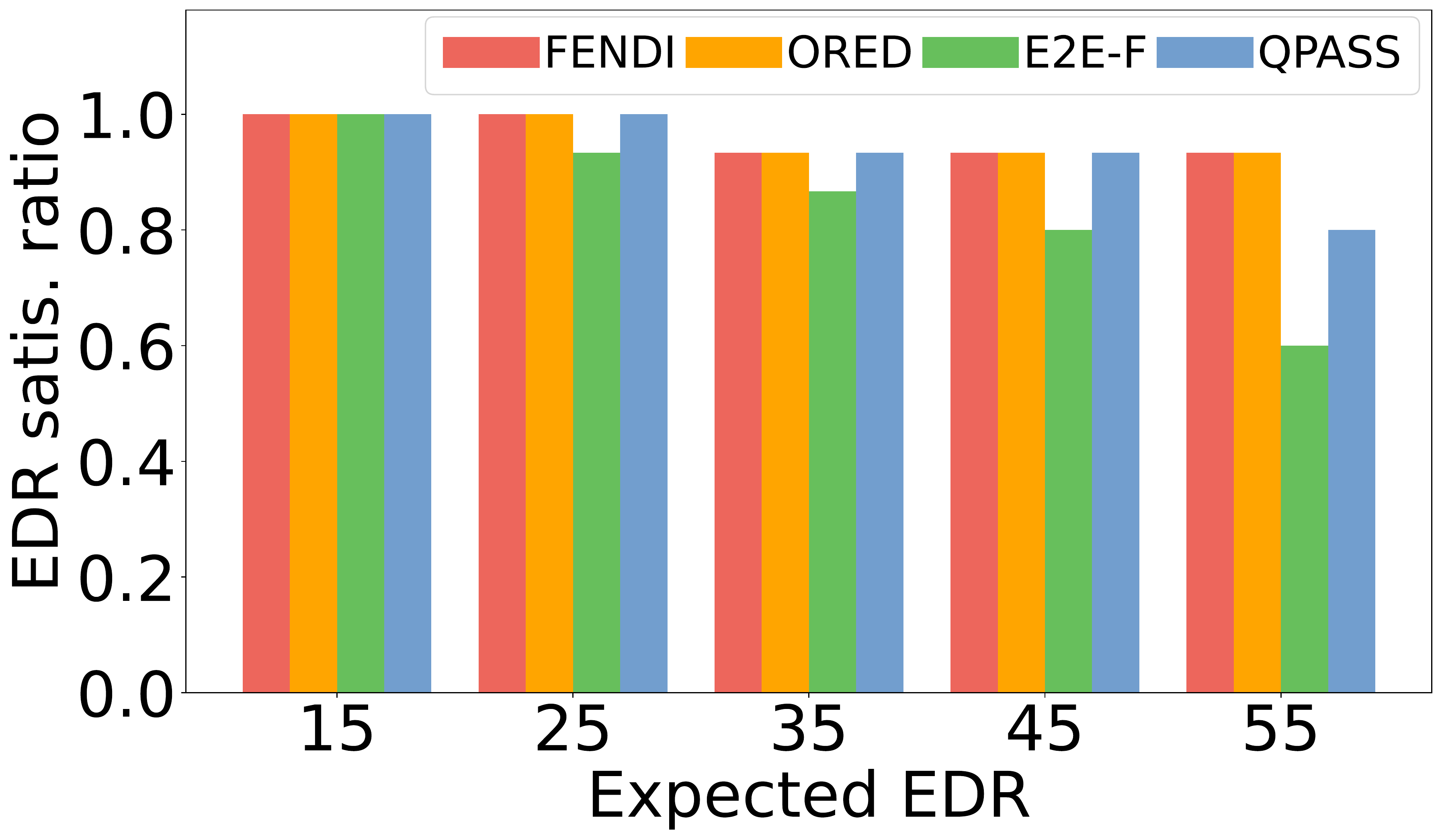}
\label{fig:algs_comp:algs_edr_satisfication_rate}}
\caption{Comparison between FENDI and state-of-the-art algorithms}
\label{fig:algs_comp}
\end{figure*}

\begin{figure*}[ht!]
\vspace{-0.5em}
\centering
\subfloat[Running time]{\includegraphics[width=0.3\textwidth]{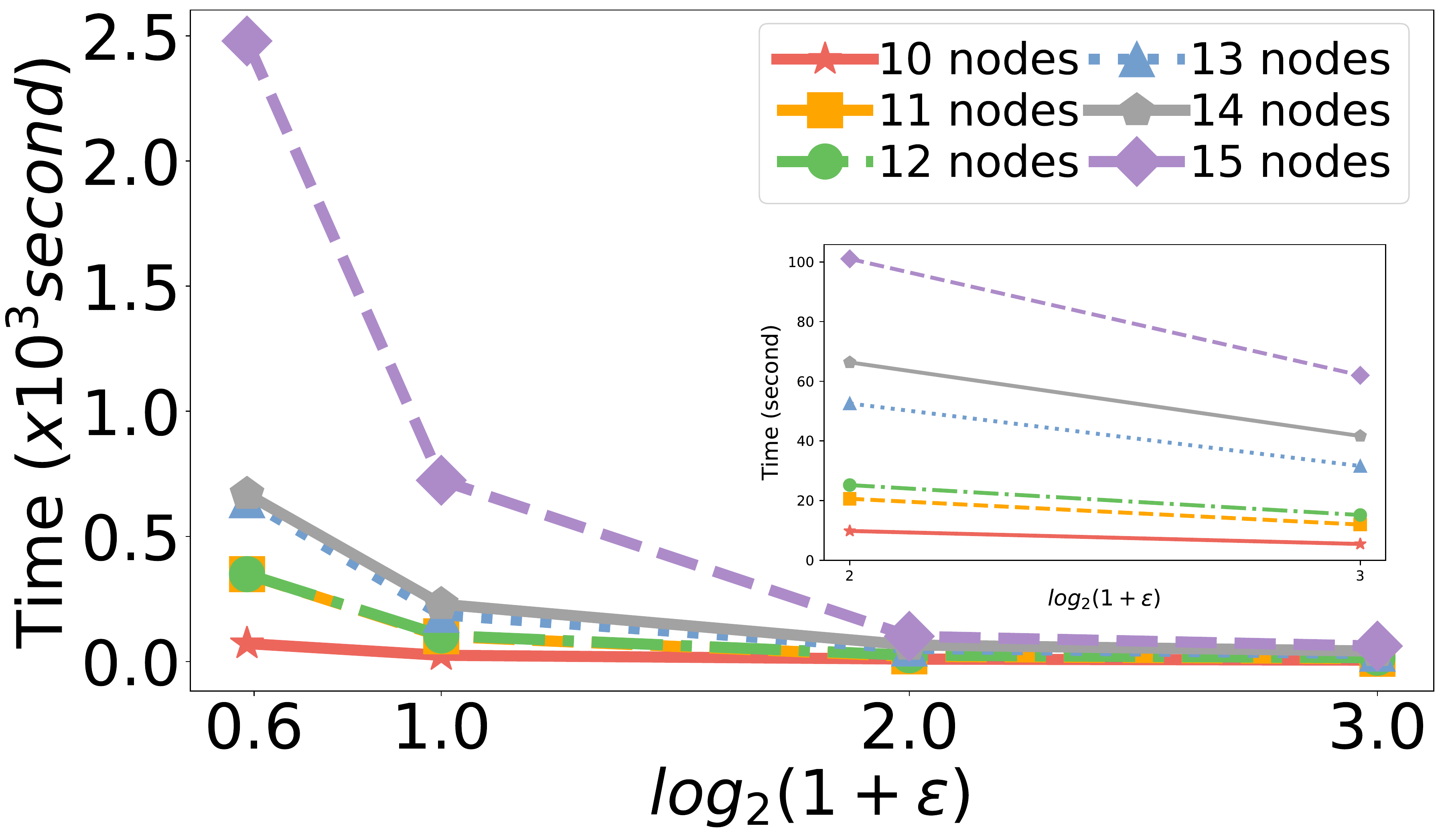}
\label{fig:appx_comp:running_time}}
\hfil
\subfloat[Minimum end-to-end fidelity]{\includegraphics[width=0.3\textwidth]{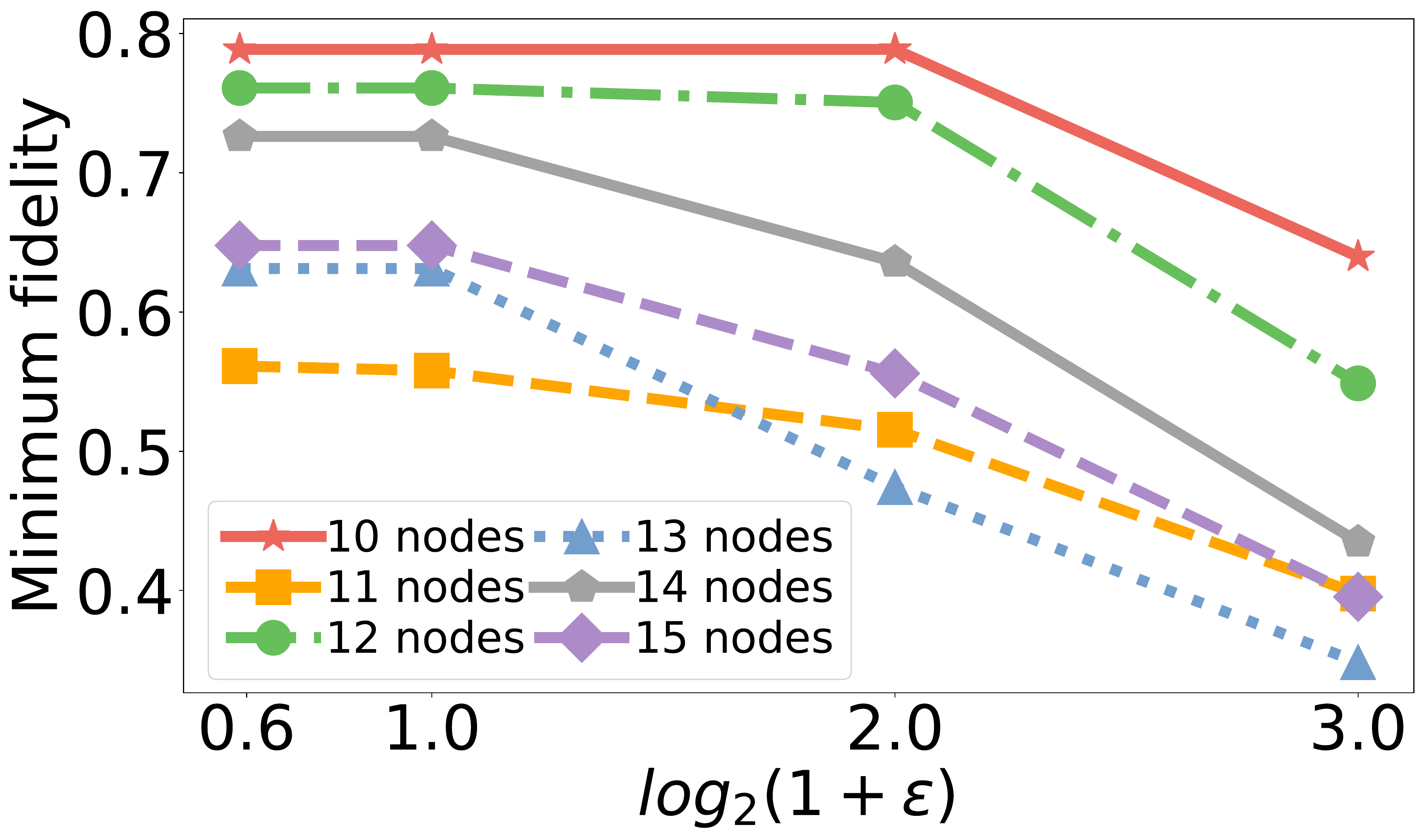}
\label{fig:appx_comp:minimum_fidelity}}
\hfil
\subfloat[Average end-to-end fidelity]{\includegraphics[width=0.3\textwidth]{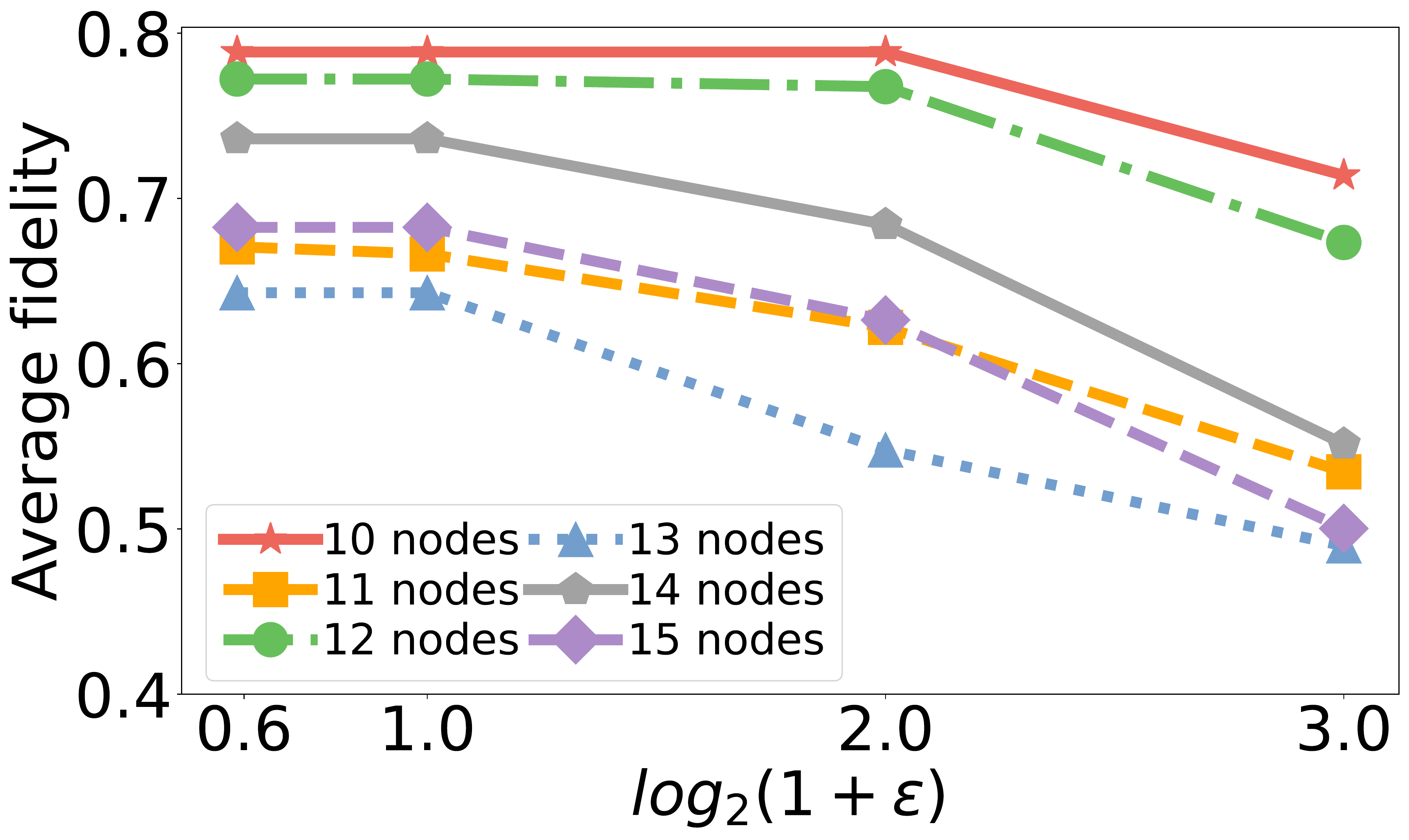}
\label{fig:appx_comp:avg_fidelity}}
\hfil
\caption{Performance and running time of FENDI with varying $\varepsilon$ and number of nodes}
\label{fig:appx_comp}
\end{figure*}

\subsection{Evaluation Results}

\subsubsection{Characterizing EDR-fidelity trade-off for single SD pair}
\noindent
We first investigate how FENDI can be used to characterize the EDR-fidelity trade-off curve for a single SD pair in a randomly generated $15$-node graph, and the result is shown in Fig.~\ref{fig:trade_off}.
We applied the $\epsilon$-constraint method~\cite{mavrotas2009effective}, varying the expected EDR bound from $1$ until the maximum value computed by ORED, and observed the maximum achievable worst-case fidelity given each expected EDR bound.
A few key observations can be made:
(i)~Even in a $15$-node network, there could be many (more than $20$) paths between a pair of nodes, leading to many strongly Pareto optimal points in the frontier.
(ii)~FENDI was able to (approximately) characterize the entire frontier from one direction, presenting many different trade-off options for entanglement distribution---each could be implemented by the post-selection-and-storage protocol.
(iii)~None of the existing algorithms could characterize the trade-off well.
Specifically, ORED could achieve the highest expected EDR, but the lowest fidelity due to using all possible paths in the network to maximize EDR.
QPASS sought to maximize EDR, but could achieve neither the maximum EDR nor the highest fidelity.
Both these methods are fidelity-agnostic, and hence could only optimize for one dimension but not the trade-off.
The fidelity-aware E2E-F was able to trade-off EDR with fidelity, but only for a very small portion of the entire trade-off curve.
The inefficacy comes from two aspects: 1) not being able to utilize all paths to achieve an arbitrary trade-off, and 2) not being able to provide guarantee for expected EDR.
In fact, most (if not all) existing algorithms are designed to optimize for a single point in the area bounded by FENDI's trade-off curve, and mostly achieve a suboptimal point strictly within the boundary.

\subsubsection{Achievable fidelity versus EDR}
\noindent
Fig.~\ref{fig:algs_comp}\subref{fig:algs_comp:algs_worst_fidelity}--\subref{fig:algs_comp:algs_average_fidelity} shows the end-to-end worst-case and average fidelity with different expected EDRs in randomly generated networks.
From Figs.~\ref{fig:algs_comp}\subref{fig:algs_comp:algs_worst_fidelity}--\subref{fig:algs_comp:algs_average_fidelity}, FENDI achieved the highest fidelity compared to all other algorithms.
For any specific expected EDR bound, the two fidelity-aware algorithms (FENDI and E2E-F) achieved significantly higher fidelity than the fidelity-agnostic ones (ORED and QPASS), demonstrating \emph{the crucial need for fidelity awareness in quantum networking}.
With increasing EDR bounds, fidelity was sacrificed to meet the EDR requirement when lower-fidelity paths were utilized.
Though both aimed to approach the optimal fidelity-EDR trade-off, the fidelity gap between FENDI and E2E-F generally increased with higher EDR bounds, demonstrating \emph{importance of our approximation guarantee}.
Note that for many tasks such as entanglement purification~\cite{bennett1996concentrating}, entanglements are regarded as non-usable when fidelity drops below $0.5$.
Fig.~\ref{fig:algs_comp}\subref{fig:algs_comp:algs_worst_fidelity} shows that to ensure minimum fidelity over $0.5$, our algorithm could achieve significantly higher expected EDR, even compared to existing fidelity-aware algorithm such as E2E-F.

\subsubsection{Capability to satisfy EDR requirement}
From Fig.~\ref{fig:algs_comp}\subref{fig:algs_comp:algs_edr_satisfication_rate}, FENDI achieved EDR satisfaction ratios on par with ORED.
This is because both algorithms explore the same EDR feasibility region, and differ only by fidelity of paths (pflows) to meet a given expected EDR bound.
Both FENDI and ORED achieved higher EDR satisfaction ratio than QPASS and E2E-F, even though E2E-F achieved similar (but still lower) fidelity compared to FENDI and higher fidelity than ORED.
There are two reasons: 1) FENDI and ORED are \emph{optimal} in terms of whether an expected EDR bound can be satisfied while E2E-F and QPASS have no such guarantee; 2) {a buffered network can achieve higher long-term EDR than a bufferless network by storing instead of discarding unused intermediate ebits.
}

\subsubsection{Performance versus running time of FPTAS}
Fig.~\ref{fig:appx_comp} shows the evaluation result for the trade-off between performance and running time for FENDI,
with varying number of nodes and accuracy parameter $\varepsilon$.
Note that despite $\varepsilon$, FENDI always achieved the same EDR satisfaction ratio as the same feasibility region of the problem was explored, and thus we omit the figure showing the EDR satisfaction ratio.
From Fig.~\ref{fig:appx_comp}\subref{fig:appx_comp:running_time}, the running time increased with number of nodes and decreased with $\varepsilon$.
From Figs.~\ref{fig:appx_comp}\subref{fig:appx_comp:minimum_fidelity} and~\ref{fig:appx_comp}\subref{fig:appx_comp:avg_fidelity}, increasing $\varepsilon$ led to fidelity reduction, matching our theoretical analysis.
However, with a relatively loose $\varepsilon$, such as when $\varepsilon\! =\! 1$, the achieved fidelity was on par with when $\varepsilon$ was set to a tight value such as $0.5$.
This shows that the theoretical guarantee tends to be over-conservative in practice, and \emph{it is reasonable to set a loose $\varepsilon$ to achieve high time efficiency with reasonable performance}.
The correlation between number of nodes and fidelity values of FENDI was weak.
This could be because, on one hand, a larger graph with more nodes could lead to more paths between each SD pair and hence increase fidelity; on the other hand, a larger graph also means it was more likely that two randomly picked nodes were further away in the graph, leading to degraded fidelity over long paths.
The potential trade-off between network size and fidelity will be explored in our future work.
%



\section{Conclusions}
\label{sec:conclusions}
\noindent 
In this paper, we studied how to characterize the entanglement distribution rate and fidelity trade-off in a general-topology quantum network with theoretical guarantee.
We derived an end-to-end fidelity model with worst-case (isotropic) noise.
We then formulated the HF-RED problem for maximizing the achievable fidelity under an expected EDR bound (modeled with an optimal entanglement flow abstraction), and proved its NP-hardness.
With a novel decomposition theorem, we developed a \emph{fully polynomial-time approximation scheme (FPTAS)} for the problem called FENDI. 
We also developed a discrete-time quantum network simulator for evaluation.
Simulation results showed the superior performance of FENDI, compared to existing entanglement routing and distribution algorithms.
%
%
%


\bibliographystyle{myIEEEtranS}


\begin{thebibliography}{10}
\providecommand{\url}[1]{#1}
\csname url@samestyle\endcsname
\providecommand{\newblock}{\relax}
\providecommand{\bibinfo}[2]{#2}
\providecommand{\BIBentrySTDinterwordspacing}{\spaceskip=0pt\relax}
\providecommand{\BIBentryALTinterwordstretchfactor}{4}
\providecommand{\BIBentryALTinterwordspacing}{\spaceskip=\fontdimen2\font plus
\BIBentryALTinterwordstretchfactor\fontdimen3\font minus
  \fontdimen4\font\relax}
\providecommand{\BIBforeignlanguage}[2]{{%
\expandafter\ifx\csname l@#1\endcsname\relax
\typeout{** WARNING: IEEEtranS.bst: No hyphenation pattern has been}%
\typeout{** loaded for the language `#1'. Using the pattern for}%
\typeout{** the default language instead.}%
\else
\language=\csname l@#1\endcsname
\fi
#2}}
\providecommand{\BIBdecl}{\relax}
\BIBdecl

\bibitem{gurobi}
\BIBentryALTinterwordspacing
``{Gurobi Optimizer},'' accessed 2022-07-25. URL:
  \url{http://www.gurobi.com/products/gurobi-optimizer}
\BIBentrySTDinterwordspacing

\bibitem{Ahuja1993}
R.~K. Ahuja, T.~L. Magnanti, and J.~B. Orlin, \emph{{Network Flows: Theory,
  Algorithms and Applictions}}.\hskip 1em plus 0.5em minus 0.4em\relax
  Prentice-Hall, 1993.

\bibitem{bayerbach2023bell}
M.~J. Bayerbach, S.~E. D’Aurelio, P.~van Loock, and S.~Barz, ``Bell-state
  measurement exceeding 50\% success probability with linear optics,''
  \emph{Science Advances}, vol.~9, no.~32, p. eadf4080, 2023.

\bibitem{bennett1996concentrating}
C.~H. Bennett, H.~J. Bernstein, S.~Popescu, and B.~Schumacher, ``Concentrating
  partial entanglement by local operations,'' \emph{Physical Review A},
  vol.~53, no.~4, p. 2046, 1996.

\bibitem{bennett2020quantum}
C.~H. Bennett and G.~Brassard, ``Quantum cryptography: Public key distribution
  and coin tossing,'' \emph{Theoretical Computer Science}, vol. 560, pp. 7--11,
  2014.

\bibitem{bennett1996mixed}
C.~H. Bennett, D.~P. DiVincenzo, J.~A. Smolin, and W.~K. Wootters,
  ``Mixed-state entanglement and quantum error correction,'' \emph{Physical
  Review A}, vol.~54, no.~5, p. 3824, 1996.

\bibitem{cacciapuoti2019quantum}
A.~S. Cacciapuoti, M.~Caleffi, F.~Tafuri, F.~S. Cataliotti, S.~Gherardini, and
  G.~Bianchi, ``{Quantum Internet: Networking challenges in distributed quantum
  computing},'' \emph{IEEE Network}, vol.~34, no.~1, pp. 137--143, 2019.

\bibitem{caleffi2018quantum}
M.~Caleffi, A.~S. Cacciapuoti, and G.~Bianchi, ``{Quantum Internet: From
  communication to distributed computing!}'' in \emph{ACM NANOCOM}, 2018, pp.
  1--4.

\bibitem{chakraborty2019distributed}
K.~Chakraborty, F.~Rozpedek, A.~Dahlberg, and S.~Wehner, ``Distributed routing
  in a quantum internet,'' \emph{arXiv preprint arXiv:1907.11630}, 2019.

\bibitem{chang2022order}
A.~Chang and G.~Xue, ``Order matters: On the impact of swapping order on an
  entanglement path in a quantum network,'' in \emph{IEEE INFOCOM WKSHPS},
  2022, pp. 1--6.

\bibitem{cicconetti2022resource}
C.~Cicconetti, M.~Conti, and A.~Passarella, ``Resource allocation in quantum
  networks for distributed quantum computing,'' \emph{arXiv preprint
  arXiv:2203.05844}, 2022.

\bibitem{dahlberg2019link}
A.~Dahlberg, M.~Skrzypczyk, T.~Coopmans, L.~Wubben, F.~Rozp{e}dek, M.~Pompili,
  A.~Stolk, P.~Pawe{\l}czak, R.~Knegjens, J.~de~Oliveira~Filho \emph{et~al.},
  ``A link layer protocol for quantum networks,'' in \emph{ACM SIGCOMM}, 2019,
  pp. 159--173.

\bibitem{Dai2020a}
W.~Dai, T.~Peng, and M.~Z. Win, ``Optimal protocols for remote entanglement
  distribution,'' in \emph{IEEE ICNC}, 2020, pp. 1014--1019.

\bibitem{Dai2020b}
------, ``Optimal remote entanglement distribution,'' \emph{IEEE Journal on
  Selected Areas in Communications}, vol.~38, no.~3, pp. 540--556, 2020.

\bibitem{Dur1999}
W.~Dur, H.-J. Briegel, J.~I. Cirac, and P.~Zoller, ``Quantum repeaters based on
  entanglement purification,'' \emph{Physical Review A}, vol.~59, pp. 169--181,
  Jan 1999.

\bibitem{elliott2002building}
C.~Elliott, ``Building the quantum network,'' \emph{New Journal of Physics},
  vol.~4, no.~1, p.~46, 2002.

\bibitem{feynman1982simulating}
R.~P. Feynman, ``Simulating physics with computers,'' \emph{International
  Journal of Theoretical Physics}, vol.~21, no. 6/7, 1982.

\bibitem{esdi}
H.~Gu, R.~Yu, Z.~Li, X.~Wang, and F.~Zhou, ``Esdi: Entanglement scheduling and
  distribution in the quantum internet,'' in \emph{IEEE ICCCN}, 2023.

\bibitem{koashi1998no}
M.~Koashi and N.~Imoto, ``No-cloning theorem of entangled states,''
  \emph{Physical Review Letters}, vol.~81, no.~19, p. 4264, 1998.

\bibitem{kozlowski2020designing}
W.~Kozlowski, A.~Dahlberg, and S.~Wehner, ``Designing a quantum network
  protocol,'' in \emph{ACM CoNEXT}, 2020, pp. 1--16.

\bibitem{kaipingxue_2022}
J.~Li, Q.~Jia, K.~Xue, D.~S.~L. Wei, and N.~Yu, ``A connection-oriented
  entanglement distribution design in quantum networks,'' \emph{IEEE
  Transactions on Quantum Engineering}, vol.~3, pp. 1--13, 2022.

\bibitem{liao2017satellite}
S.-K. Liao, W.-Q. Cai, W.-Y. Liu, L.~Zhang, Y.~Li, J.-G. Ren, J.~Yin, Q.~Shen,
  Y.~Cao, Z.-P. Li \emph{et~al.}, ``Satellite-to-ground quantum key
  distribution,'' \emph{Nature}, vol. 549, no. 7670, pp. 43--47, 2017.

\bibitem{ma2021a}
Y.~Ma, Y.-Z. Ma, Z.-Q. Zhou, C.-F. Li, and G.-C. Guo, ``One-hour coherent
  optical storage in an atomic frequency comb memory,'' \emph{Nature
  Communications}, vol.~12, no.~1, pp. 1--6, 2021.

\bibitem{mavrotas2009effective}
G.~Mavrotas, ``Effective implementation of the $\varepsilon$-constraint method
  in multi-objective mathematical programming problems,'' \emph{Applied
  mathematics and computation}, vol. 213, no.~2, pp. 455--465, 2009.

\bibitem{Miettinen1998}
K.~Miettinen, \emph{{Nonlinear Multiobjective Optimization}}, ser.
  International Series in Operations Research {\&} Management Science.\hskip
  1em plus 0.5em minus 0.4em\relax Boston, MA: Springer US, 1998, vol.~12.

\bibitem{misra2009polynomial}
S.~Misra, G.~Xue, and D.~Yang, ``Polynomial time approximations for multi-path
  routing with bandwidth and delay constraints,'' in \emph{IEEE INFOCOM}, 2009,
  pp. 558--566.

\bibitem{Muralidharan2016}
S.~Muralidharan, L.~Li, J.~Kim, N.~L{\"{u}}tkenhaus, M.~D. Lukin, and L.~Jiang,
  ``{Optimal architectures for long distance quantum communication},''
  \emph{Scientific Reports}, vol.~6, no.~1, p. 20463, feb 2016.

\bibitem{panigrahy2022capacity}
N.~K. Panigrahy, T.~Vasantam, D.~Towsley, and L.~Tassiulas, ``On the capacity
  region of a quantum switch with entanglement purification,'' in \emph{IEEE
  INFOCOM}, 2023.

\bibitem{pant2019routing}
M.~Pant, H.~Krovi, D.~Towsley, L.~Tassiulas, L.~Jiang, P.~Basu, D.~Englund, and
  S.~Guha, ``Routing entanglement in the quantum internet,'' \emph{npj Quantum
  Information}, vol.~5, no.~1, pp. 1--9, 2019.

\bibitem{peev2009secoqc}
M.~Peev, C.~Pacher, R.~All{\'e}aume, C.~Barreiro, J.~Bouda, W.~Boxleitner,
  T.~Debuisschert, E.~Diamanti, M.~Dianati, J.~Dynes \emph{et~al.}, ``The
  secoqc quantum key distribution network in vienna,'' \emph{New Journal of
  Physics}, vol.~11, no.~7, p. 075001, 2009.

\bibitem{photon-loss}
S.~Pirandola, R.~Laurenza, C.~Ottaviani, and L.~Banchi, ``Fundamental limits of
  repeaterless quantum communications,'' \emph{Nature communications}, vol.~8,
  no.~1, p. 15043, 2017.

\bibitem{pouryousef2022quantum}
S.~Pouryousef, N.~K. Panigrahy, and D.~Towsley, ``A quantum overlay network for
  efficient entanglement distribution,'' \emph{arXiv preprint
  arXiv:2212.01694}, 2022.

\bibitem{sasaki2011field}
M.~Sasaki, M.~Fujiwara, H.~Ishizuka, W.~Klaus, K.~Wakui, M.~Takeoka, S.~Miki,
  T.~Yamashita, Z.~Wang, A.~Tanaka \emph{et~al.}, ``{Field test of quantum key
  distribution in the Tokyo QKD Network},'' \emph{Optics Express}, vol.~19,
  no.~11, pp. 10\,387--10\,409, 2011.

\bibitem{schoute2016shortcuts}
E.~Schoute, L.~Mancinska, T.~Islam, I.~Kerenidis, and S.~Wehner, ``Shortcuts to
  quantum network routing,'' \emph{arXiv preprint arXiv:1610.05238}, 2016.

\bibitem{shi2020concurrent}
S.~Shi and C.~Qian, ``Concurrent entanglement routing for quantum networks:
  Model and designs,'' in \emph{ACM SIGCOMM}, 2020, pp. 62--75.

\bibitem{singh2021quantum}
A.~Singh, K.~Dev, H.~Siljak, H.~D. Joshi, and M.~Magarini, ``Quantum
  internet—applications, functionalities, enabling technologies, challenges,
  and research directions,'' \emph{IEEE Communications Surveys \& Tutorials},
  vol.~23, no.~4, pp. 2218--2247, 2021.

\bibitem{van2013designing}
R.~Van~Meter and J.~Touch, ``Designing quantum repeater networks,'' \emph{IEEE
  Communications Magazine}, vol.~51, no.~8, pp. 64--71, 2013.

\bibitem{vardoyan2021stochastic}
G.~Vardoyan, S.~Guha, P.~Nain, and D.~Towsley, ``On the stochastic analysis of
  a quantum entanglement distribution switch,'' \emph{IEEE Transactions on
  Quantum Engineering}, vol.~2, pp. 1--16, 2021.

\bibitem{victora2020purification}
M.~Victora, S.~Krastanov, A.~S. de~la Cerda, S.~Willis, and P.~Narang,
  ``Purification and entanglement routing on quantum networks,'' \emph{arXiv
  preprint arXiv:2011.11644}, 2020.

\bibitem{waxman}
B.~M. Waxman, ``Routing of multipoint connections,'' \emph{IEEE Journal on
  Selected Areas in Communications}, vol.~6, no.~9, pp. 1617--1622, 1988.

\bibitem{Xia2020}
Y.~Xia, W.~Li, W.~Clark, D.~Hart, Q.~Zhuang, and Z.~Zhang, ``{Demonstration of
  a Reconfigurable Entangled Radio-Frequency Photonic Sensor Network},''
  \emph{Physical Review Letters}, vol. 124, no.~15, p. 150502, apr 2020.

\bibitem{zhao2023scheduling}
L.~Yang, Y.~Zhao, L.~Huang, and C.~Qiao, ``Asynchronous entanglement
  provisioning and routing for distributed quantum computing,'' in \emph{IEEE
  INFOCOM}, 2023.

\bibitem{Ye1991}
Y.~Ye and P.~M. Pardalos, ``A class of linear complementarity problems solvable
  in polynomial time,'' \emph{Linear Algebra and Its Applications}, vol. 152,
  pp. 3--17, 1991.

\bibitem{yin2017satellite}
J.~Yin, Y.~Cao, Y.-H. Li, S.-K. Liao, L.~Zhang, J.-G. Ren, W.-Q. Cai, W.-Y.
  Liu, B.~Li, H.~Dai \emph{et~al.}, ``Satellite-based entanglement distribution
  over 1200 kilometers,'' \emph{Science}, vol. 356, no. 6343, pp. 1140--1144,
  2017.

\bibitem{zeng2022multi}
Y.~Zeng, J.~Zhang, J.~Liu, Z.~Liu, and Y.~Yang, ``Multi-entanglement routing
  design over quantum networks,'' in \emph{IEEE INFOCOM}, 2022.

\bibitem{zhao2021redundant}
Y.~Zhao and C.~Qiao, ``Redundant entanglement provisioning and selection for
  throughput maximization in quantum networks,'' in \emph{IEEE INFOCOM}, 2021,
  pp. 1--10.

\bibitem{zhao2022e2e}
Y.~Zhao, G.~Zhao, and C.~Qiao, ``{E2E fidelity aware routing and purification
  for throughput maximization in quantum networks},'' in \emph{IEEE INFOCOM},
  2022.

\end{thebibliography}



\appendices
\section{Data Plane Protocol for FENDI}

\newcommand{\set}{{buffer}}

\noindent
Given a solution output by a central quantum network controller running Algorithm~\ref{a:fptas}, 
we design an extension of the protocol in~\cite{Dai2020a} to achieve the expected EDR and guarantee that all generated ebits have end-to-end fidelity of at least $\Upsilon_{st}$.

Specifically, after the computation, the quantization $\theta$ and the final quantized path length bound $Z_{\text{UB}}$ are distributed to each quantum repeater along with the solution.
For every enode $m{{{{{}}}}}n$, both nodes maintain \emph{\textbf{input \set{}s}} $\mathcal{E}_{m{{{{{}}}}}n/z}$ for every value $z = 1, 2, \dots, Z_{\text{UB}}$ where $I(mn/z) > 0$.
$\mathcal{E}_{m{{{{{}}}}}n/z}$ stores the ebits generated between $m{{{{{}}}}}n$ with a specific range of fidelity values represented by a quantized length $z$.
They also maintain \emph{\textbf{output \set{}s}} $\mathcal{D}^{m{{{{{}}}}}n/z}_{m{{{{{}}}}}k/z'}$ respectively for every $k \ne m, n$ and $z' > z$ where $f^{mn/z}_{mk/z'} > 0$, which stores the ebits that will be contributed to generating ebits between other pairs with other fidelity values.
Note that the number and sizes of buffers at each node may be dynamically adjusted by allocating the available quantum memories. 

To execute the protocol, each link $mn \in E$ will continuously generate $c_{mn} \cdot g_{m{{{{{}}}}}n}$ elementary ebits.
Once successfully generated, these ebits are added to the \set{} $\mathcal{E}_{m{{{{{}}}}}n/z}$ where $z \!=\! \zeta_{mn}^\theta \!=\! \lfloor -\log(W_{mn}) \theta \rfloor \!+\! 1$.
Simultaneously, whenever an ebit is added to $\mathcal{E}_{m{{{{{}}}}}n/z}$ for any $z$, the two end points will jointly toss a random coin, and move the ebit from $\mathcal{E}_{m{{{{{}}}}}n/z}$ to $\mathcal{D}^{m{{{{{}}}}}n/z}_{m{{{{{}}}}}k/z'}$ or $\mathcal{D}^{m{{{{{}}}}}n/z}_{k{{{{{}}}}}n/z'}$ with the following probabilities:
\begin{align*}
    \smallfont
    \Pr[\text{move to }\mathcal{D}^{m{{{{{}}}}}n/z}_{m{{{{{}}}}}k/z'}] = \frac{f^{m{{{{{}}}}}n/z}_{m{{{{{}}}}}k/z'}}{ \sum_{z''}\sum_{k}{ ( f^{m{{{{{}}}}}n/z}_{m{{{{{}}}}}k/z''} + f^{m{{{{{}}}}}n/z}_{k{{{{{}}}}}n/z''} ) } };    \\
    \Pr[\text{move to }\mathcal{D}^{m{{{{{}}}}}n/z}_{k{{{{{}}}}}n/z'}] = \frac{f^{m{{{{{}}}}}n/z}_{k{{{{{}}}}}n/z'}}{ \sum_{z''}\sum_{k}{ ( f^{m{{{{{}}}}}n/z}_{m{{{{{}}}}}k/z''} + f^{m{{{{{}}}}}n/z}_{k{{{{{}}}}}n/z''} ) } }.   
\end{align*}
Finally, each node $k$ will be checking if for any $m{{{{{}}}}}n$, there exists $z_1, z_2, z_3$ such that 
\begin{enumerate}
    \item $z_1 + z_2 + \zeta_k^\theta = z_3$;
    \item $f^{m{{{{{}}}}}k/z_1}_{m{{{{{}}}}}n/z_3} = f^{k{{{{{}}}}}n/z_2}_{m{{{{{}}}}}n/z_3} > 0$; and
    \item $\mathcal{D}^{m{{{{{}}}}}k/z_1}_{m{{{{{}}}}}n/z_3} \ne \emptyset$, and $\mathcal{D}^{k{{{{{}}}}}n/z_2}_{m{{{{{}}}}}n/z_3} \ne \emptyset$.
\end{enumerate}
For each such a case, node $k$ locally performs swapping between each pair of ebits in $\mathcal{D}^{m{{{{{}}}}}k/z_1}_{m{{{{{}}}}}n/z_3}$ and $\mathcal{D}^{k{{{{{}}}}}n/z_2}_{m{{{{{}}}}}n/z_3}$ respectively.
Upon success, the ebit will then be added to $\mathcal{E}_{m{{{{{}}}}}n/z_3}$ by $m$ and $n$.
The source and destination will keep all ebits received in $\mathcal{E}_{s{{{{{}}}}}t/z}$ for any $z$.
All the above processes can be parallel and asynchronous.  %
The strong network-wide synchronization requirement in traditional time-slotted entanglement routing protocols is thus relaxed.
By an induction proof similar to the one in~\cite{Dai2020a} which we omit due to page limit, this protocol is guaranteed to achieve a long-term EDR of at least $\Delta_{st}$ and an end-to-end fidelity of at least $\Upsilon_{st}$ output by the algorithm. 

\textbf{\emph{Remark:}}
One implicit assumption not mentioned in~\cite{Dai2020a} is that the proposed protocol requires perfect quantum memories to provide the guaranteed fidelity, and sufficiently large memories to achieve the full expected EDR.
These assumptions are somewhat unrealistic under the current technologies.
Hence, the computed EDR and fidelity both serve as upper bounds on the actual values that can be achieved by near-term devices.
Though it is fairly well agreed that large-scale long-lived quantum memories will be an integral part of quantum networks in the future, especially with recent breakthroughs in optical memory devices with more than 1-hour coherence time~\cite{ma2021a}.

On the other hand, we believe even establishing (tight) bounds on the achievable EDR and fidelity is still very useful for near-term quantum network design, such as when comparing different network topologies and parameters, or comparing practical protocol design with these theoretical upper bounds.
Furthermore, we have also preliminarily tested the performance of the buffered protocol above with limited buffer space, and found that it can still maintain an EDR close to the theoretical bound with a relatively small buffer size---such as equal to the capacity of each link.
While out of the scope of the current paper which focuses on computing the theoretical bounds, we believe smart buffer management can further reduce the buffer size and increase achievable EDR and fidelity, which we will explore in our future work.
%
%



%
\begin{IEEEbiography}[{\includegraphics[width=1in,height=1.25in,clip,keepaspectratio]{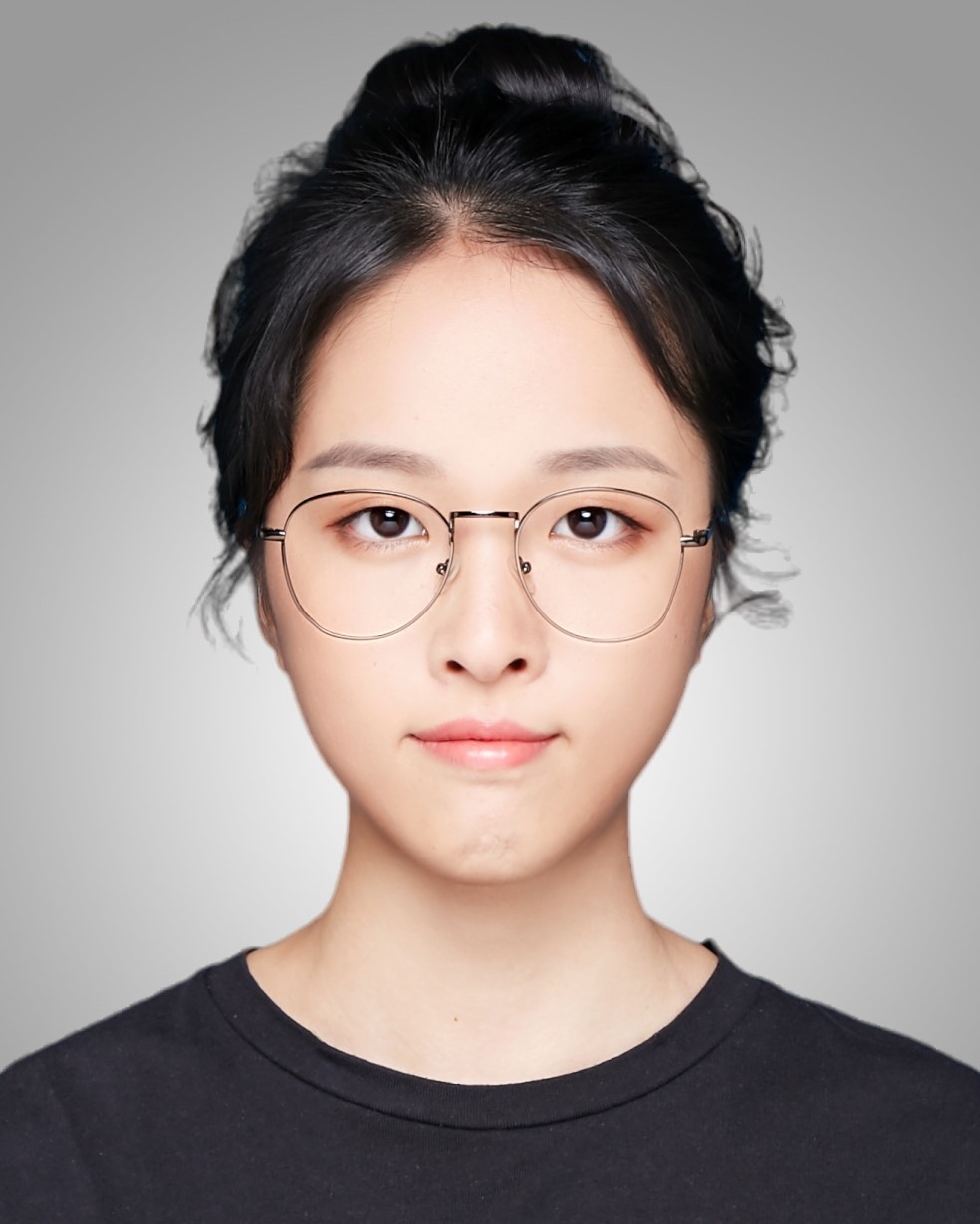}}]{Huayue Gu}(Student Member 2021) received her M.S. degree from the University of California, Riverside, CA, USA, in 2021. Currently, she is a Ph.D. student in the Computer Science department at North Carolina State University. Her research interests are quantum networking, quantum communication, data analytics, etc.
\end{IEEEbiography}



\begin{IEEEbiography}[{\includegraphics[width=1in,height=1.25in,clip,keepaspectratio]{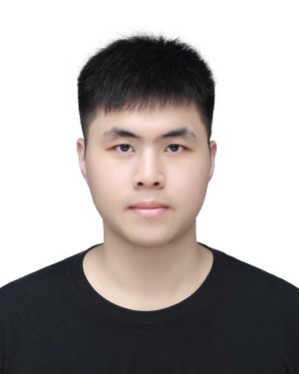}}]{Zhouyu Li}
(Student Member 2021) received his B.E. degree from Central South University, Changsha, China, in 2019 and his M.S. degree from Georgia Institute of Technology, Atlanta, U.S., in 2020. Currently, he is a Ph.D. student of Computer Science at North Carolina State University. His research interests include privacy, cloud/edge computing, network routing, etc.
\end{IEEEbiography}



\begin{IEEEbiography}
[{\includegraphics[width=1in,height=1.25in,clip,keepaspectratio]{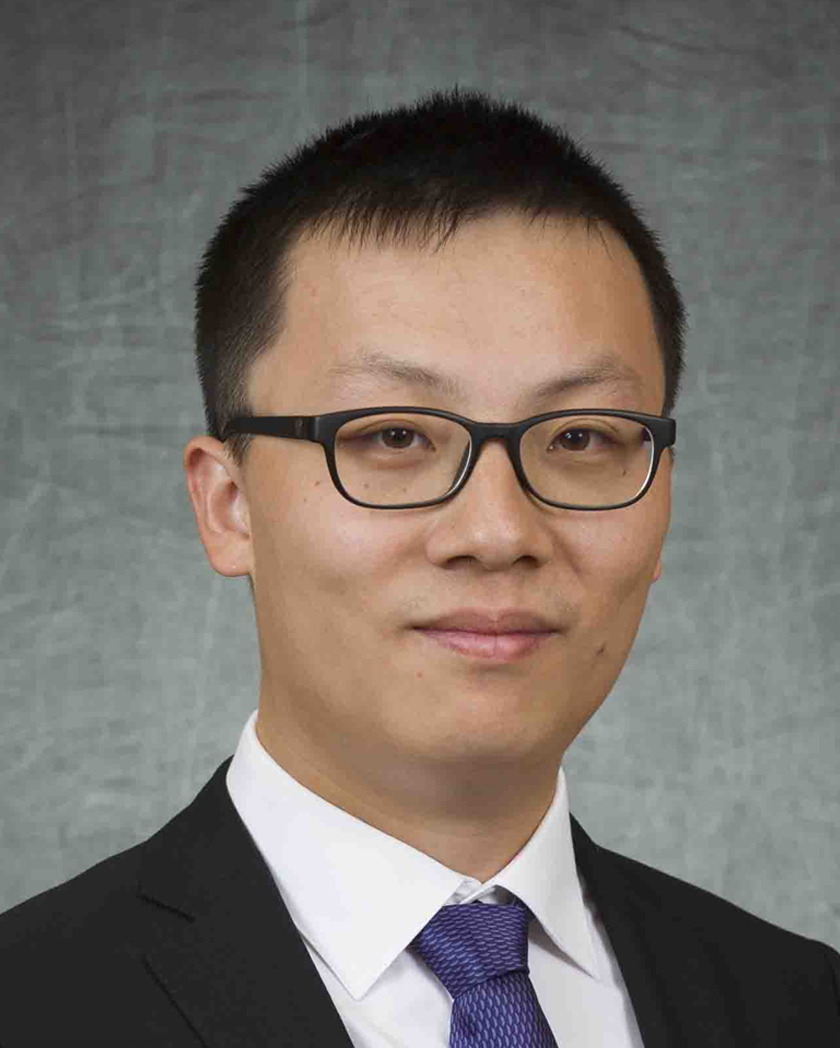}}]
{Ruozhou Yu} (Student Member 2013, Member 2019, Senior Member 2021) is an Assistant Professor of Computer Science at NC State University, USA. He received his Ph.D. degree (2019) in Computer Science from Arizona State University, USA. His interests include quantum networking, edge computing, algorithms and optimization, distributed learning, and security and privacy. He has served on the Organizing Committees of IEEE INFOCOM 2022-2024 and IEEE IPCCC 2020-2024, as a TPC Track Chair for IEEE ICCCN 2023, and as TPC members of IEEE INFOCOM 2020-2024 and ACM Mobihoc 2023. He is an Area Editor for Elsevier Computer Networks. He received the NSF CAREER Award in 2021.
\end{IEEEbiography}



\begin{IEEEbiography}[{\includegraphics[width=1in,height=1.25in,clip,keepaspectratio]{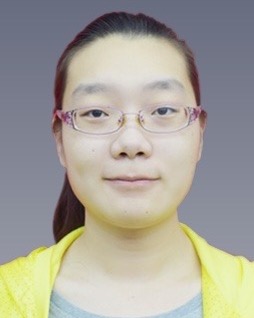}}]{Xiaojian Wang}(Student Member 2021) received her B.E. degree from Taiyuan University of Technology, China, in 2017 and received her M.S. degree in Computer Science from University of West Florida, FL, USA and Taiyuan University of Technology, China, in 2020. She is now a Ph.D. student in the department of Computer Science, College of Engineering at North Carolina State University. Her research interests include payment channel network, security, blockchain. 
\end{IEEEbiography}



\begin{IEEEbiography}[{\includegraphics[width=1in,height=1.25in,clip,keepaspectratio]{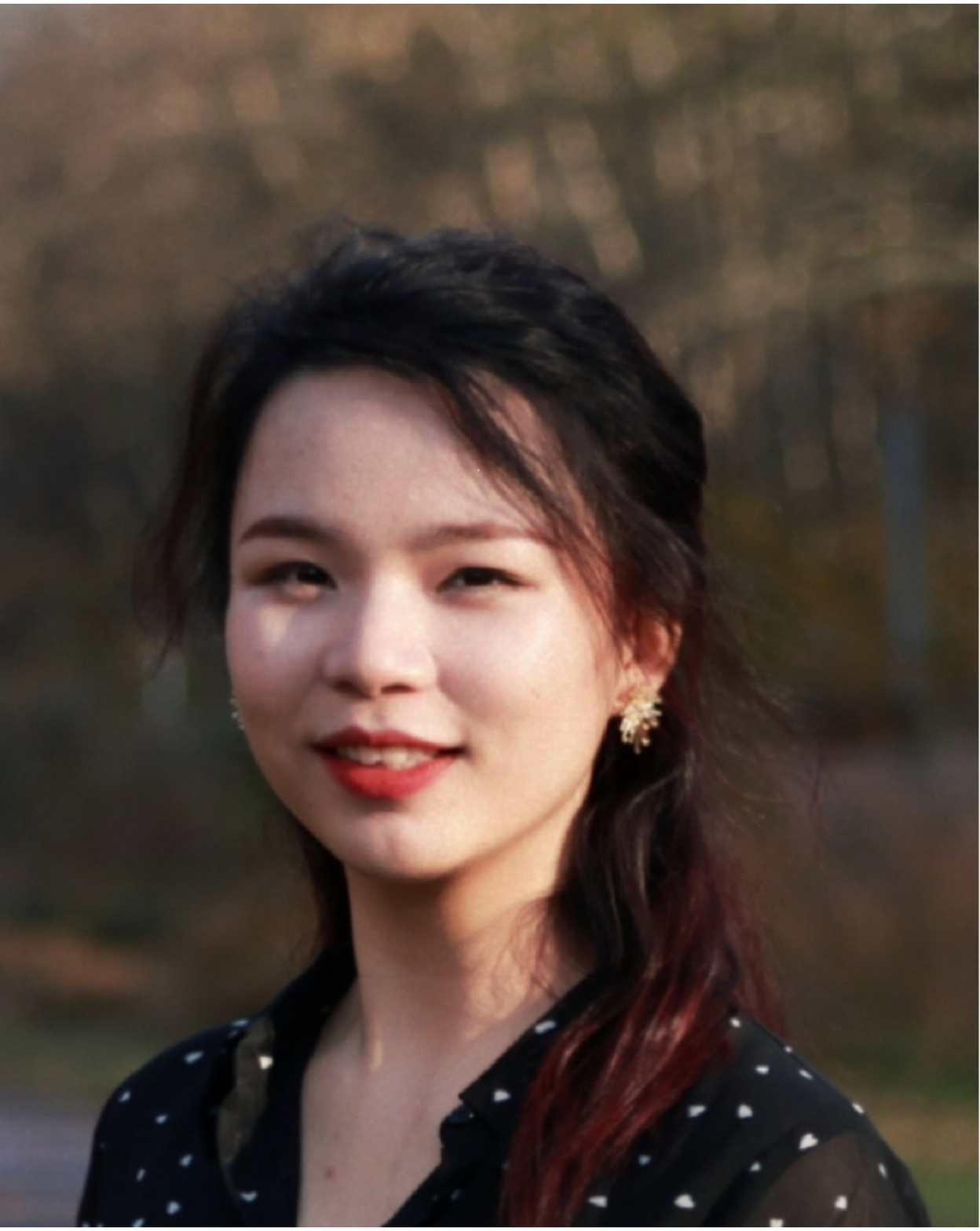}}]{Fangtongzhou}(Student Member 2021) received her B.E. degree (2018) in Electrical Engineering and Automation from Harbin Institute of Technology, Harbin, China and M.S. degree (2020) in Electrical Engineering from Texas A\&M University, College Station, Texas, USA. Currently she is a Ph.D candidate in the School of Computer Science at North Carolina State University. Her research interests include machine learning in computer networking, like federated learning, reinforcement learning for resource provisioning. 
\end{IEEEbiography}



\begin{IEEEbiography}
[{\includegraphics[width=1in,height=1.25in,clip,keepaspectratio]{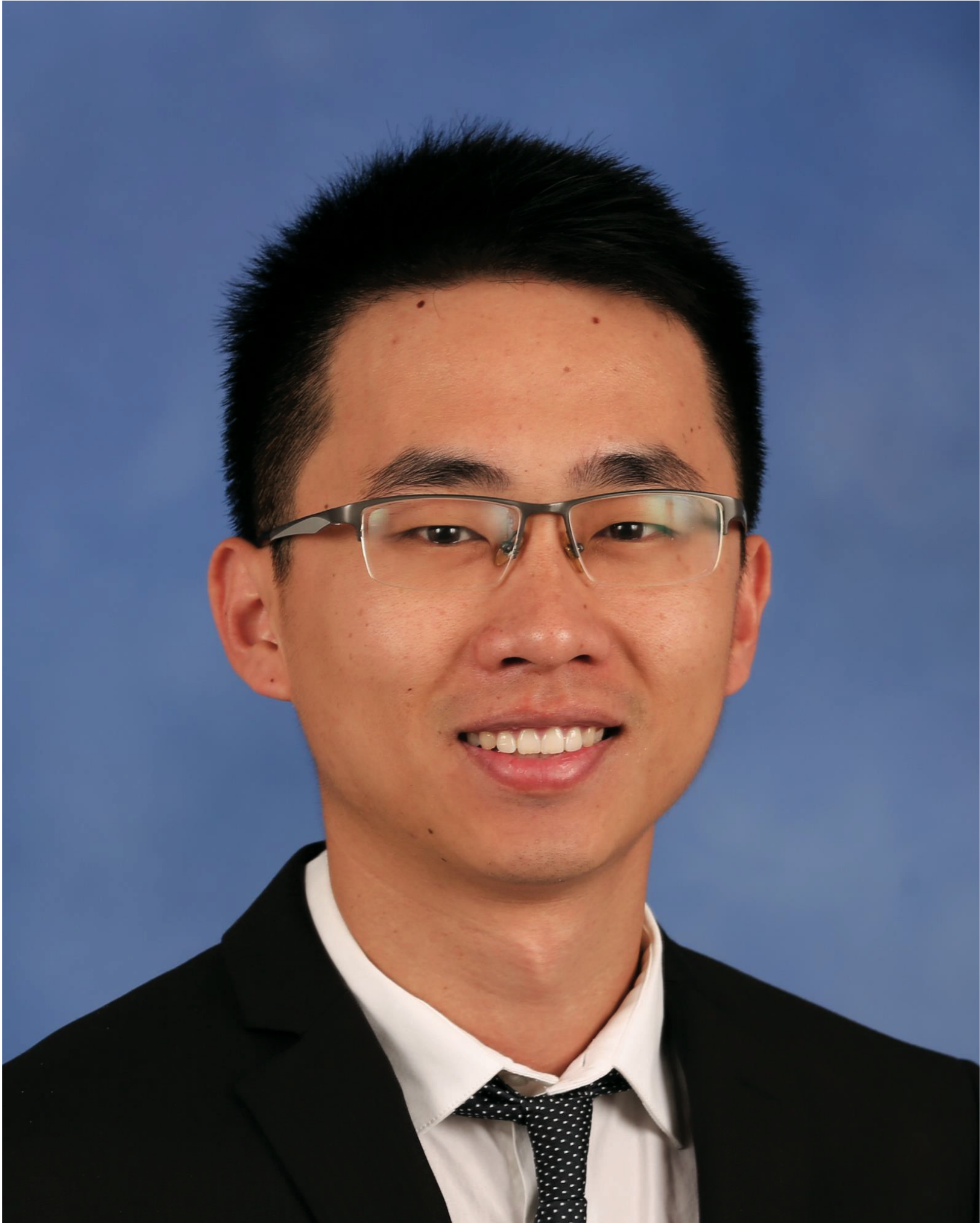}}]
{Jianqing Liu}(Member 2018) is currently an Assistant Professor of Computer Science at NC State University. He received the Ph.D. degree from The University of Florida in 2018 and the B.S. degree from University of Electronic Science and Technology of China in 2013. His research interest is wireless communications and networking, security and privacy. He received the US NSF CAREER Award in 2021. He also received several best paper awards including 2018 Best Journal Paper Award from IEEE TCGCC.
\end{IEEEbiography}



\begin{IEEEbiography}[{\includegraphics[width=1in,height=1.25in,clip,keepaspectratio]{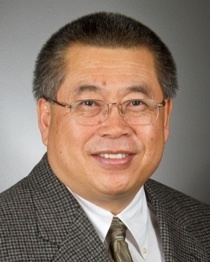}}]
{Guoliang Xue} (Member 1996, Senior Member 1999, Fellow 2011) is
a Professor of Computer Science in the School of Computing and Augmented
Intelligence at Arizona State University.
His research interests span the areas of
Internet-of-things,
cloud/edge/quantum computing and networking,
crowdsourcing and truth discovery,
QoS provisioning and network optimization,
security and privacy,
optimization and machine learning.
He received the IEEE Communications Society William R. Bennett Prize in 2019.
He is an Associate Editor of IEEE Transactions on Mobile Computing,
as well as a member of the Steering Committee of this journal.
He served on the editorial boards of
IEEE/ACM Transactions on Networking
and
IEEE Network Magazine,
as well as the Area Editor of
IEEE Transactions on Wireless Communications, overseeing 13 editors
in the Wireless Networking area.
He has served as VP-Conferences of the IEEE Communications Society.
He is the Steering Committee Chair of IEEE INFOCOM.
\end{IEEEbiography}



\end{document}